\documentclass[10pt,journal]{IEEEtran}

\usepackage{srcltx}
\usepackage[psamsfonts]{amsfonts}
\usepackage{amsmath}
\usepackage{amssymb}
\usepackage{graphics}
\usepackage{psfrag}
\usepackage{epsfig}
\usepackage{subfigure}
\usepackage{array}
\usepackage{makeidx}  
\usepackage{amssymb}
\usepackage{bbm}
\usepackage{comment}
\usepackage{amsfonts}
\usepackage{hhline}
\usepackage{mathrsfs}
\usepackage{srcltx}
\usepackage{eufrak}
\usepackage{yfonts}
\usepackage{pifont}
\usepackage{footnote}
\usepackage{fancybox}
\usepackage{hyperref}
\usepackage{algorithm}
\usepackage{algorithmic}
\usepackage{cite}
\usepackage{url}
\usepackage{comment}
\interdisplaylinepenalty=2500 \algsetup{indent=2em}

\newcommand{\ora}{\overrightarrow}

\newcommand{\E}{\ora{E}}

\newcommand{\ve}{\varepsilon}

\newcommand{\vp}{\varphi}
\newcommand{\f}{\mathbbm{F}}
\newcommand{\tee}{\zeta}

\newcommand{\mce}{\mathcal{E}}
\newcommand{\mcs}{\mathcal{S}}
\newcommand{\mcm}{\mathcal{M}}
\newcommand{\mcr}{\mathbf{r}}
\newcommand{\mcW}{\mathcal{W}}
\newcommand{\mcp}{\mathbf{p_b}}
\newcommand{\bfn}{\mathbf{n}}
\newcommand{\bfq}{\mathbf{q}}
\newcommand{\bfnt}{\mathbf{\tilde{n}}}
\newcommand{\nt}{{\tilde{n}}}
\newcommand{\bfs}{\mathbf{s}}
\newcommand{\mcc}{\mathcal{C}}
\newcommand{\lsp}{\textrm{span}}
\newcommand{\lss}{\mathscr{S}}
\newcommand{\rk}{\textrm{rank}}
\newcommand{\ra}{\rightarrow}
\newcommand{\fq}{\mathbbm{F}_q}
\newcommand{\ol}{\overline}
\newcommand{\ul}{\underline}
\newcommand{\mbb}{\mathbbm}
\newcommand{\mc}{\mathcal}
\newcommand{\mg}{\mathbbm{G}}

\newtheorem{thm}{Theorem}
\newtheorem{cor}{Corollary}

\newtheorem{lem}{Lemma}

\begin{document}
\title{\huge{Throughput and Latency in Finite-Buffer Line Networks}}
\author{Badri N. Vellambi, Nima Torabkhani and Faramarz Fekri}
\pagestyle{plain}
\maketitle
\thispagestyle{plain}
\begin{abstract}
This work investigates the effect of finite buffer sizes on the throughput capacity and packet delay of line networks with packet erasure links that have perfect feedback. These performance measures are shown to be linked to the stationary distribution of an underlying irreducible Markov chain that models the system exactly. Using simple strategies, bounds on the throughput capacity are derived. The work then presents two iterative schemes to approximate the steady-state distribution of node occupancies by decoupling the chain to smaller queueing blocks. These approximate solutions are used to understand the effect of buffer sizes on throughput capacity and the distribution of packet delay. Using the exact modeling for line networks, it is shown that the throughput capacity is unaltered in the absence of hop-by-hop feedback provided packet-level network coding is allowed. Finally, using simulations, it is confirmed that the proposed framework yields accurate estimates of the throughput capacity and delay distribution and captures the vital trends and tradeoffs in these networks.
\end{abstract}

\section{Introduction}
In networks, packets have to be routed between nodes through a series of intermediate relay nodes. Each intermediate node in the network may receive packets via multiple data streams that are routed simultaneously from their source nodes to their respective destinations. In such conditions, packets may have to be stored at intermediate nodes for transmission at a later time. If buffers are unlimited, intermediate nodes need not have to reject or drop arriving packets. However, in practice, buffers are limited in size. Although a large buffer size is preferred to minimize packet drops, large buffers have an adverse effect on the \emph{latency}, i.e., the delay experienced by packets stored in the network. Further, using larger buffer sizes at intermediate nodes would also result in secondary practical issues such as increased memory-access latency. Though our work is motivated by such concerns, our work is far from modeling realistic conditions. This work modestly aims at providing a theoretical framework to understand the fundamental limits of single information flow in finite-buffer line networks and investigates the tradeoffs between throughput, packet delay and buffer size.

The problem of computing capacity\footnote{In this work, we use capacity to refer to the throughput capacity, {i.e.}, the supremum of all rates of information flow achievable by any coding scheme.} and designing efficient coding schemes for lossy wired and wireless networks has been widely studied~\cite{AmirDana:capacity, AmirDanaAllerton, PakzadF05, NCAlgApproch, YeungNCJrnl}. However, the study of capacity of networks with finite buffer sizes has been limited. This can be attributed solely to the fact that analysis of finite buffer systems are generally more challenging. With the advent of network coding as an elegant and effective tool for attaining optimum network performance, the interest in finite-buffer networks has increased~\cite{YeungNCJrnl, LunR05, LunM05, LUN04}.

The problem of studying lossy networks with finite buffers has been investigated in the area of queueing theory under a different but similar framework. The queueing theory framework attempts to model packets in a network as customers, the delay due to packet loss over links as service times in the nodes, and the buffer size at intermediate nodes as the queue size. Further, the phenomenon of packet overflow in communications network is modeled by blocking (commonly known as \emph{type II} or \emph{blocking after service}) in queueing networks~\cite{HJPBook}. However, this packet-customer equivalence fails in general network topologies due to the following reason. When the communications network contains multiple disjoint paths from the source to the destination, the source node can choose to duplicate packets on multiple paths to minimize delay. This replicating strategy cannot be captured directly in the customer-server based queueing model. Therefore, the queueing framework cannot be directly applied to study packet traffic in general communications networks. However, queueing theory offers solid foundation for studying buffer occupancies and packet flow traffic in line networks. There has been extensive study in queueing theory literature on the behavior of open tandem queues, which are analogous to line networks~\cite{Tayfur_QT1, Tayfur_QT2, Brandwajn_QT3, SO_QT4, Yves_QT5, Serfozo_QT0}. However, approaches from queueing theory literature predominantly consider a continuous-time model for arrival and departure of customers/packets. In this work, we consider a discrete-time model for packet arrival and departure processes by lumping time into epochs. This model is similar to those in~\cite{VellambiITW, NetCod:Lun}.

The broad contributions of this paper can be summarized as follows. The bulk of this work operates under the assumption of perfect hop-by-hop feedback. We present a Markov-chain based model for exact analysis of line networks. The capacity of a line network is shown to be related to the steady-state distribution of an underlying chain, whose state space grows exponentially in the number of hops in the network. Simple assumptions of renewalness of intermediate packet processes are employed to estimate the capacity of such networks. The estimates are exact for two-hop networks. However, the estimates extend the results of~\cite{NetCod:Lun} to networks of any number of hops and buffer sizes of intermediate nodes. Using the estimates, the profile of packet delay is derived and studied. Using the exact Markov chain model in conjunction with network coding, it is shown that the throughput capacity is not affected by the absence of feedback in line networks. This result is similar to the information-theoretic result that feedback does not increase capacity of point-to-point channels~\cite{Shannon_Feedback}. Finally, simulations reveal that our estimates closely predict the trends and tradeoffs between hop-length, buffer size, latency, and throughput in these networks.

This paper is organized as follows. First, we present the formal definition of the problem and the network model in Section~\ref{FB-sec1}. Next, we present our framework for analyzing capacity of finite-buffer line networks in Section~\ref{FB-sec2}. The proposed Markovian framework is then employed to investigate packet delay in Section~\ref{FB-Delay}. We compare our analytical results with simulations in Section~\ref{FB-sec4} and conclude it with a brief discussion on the inter-dependence of buffer usage, capacity and delay. Finally, Section~\ref{Conc} concludes the paper.

\section{Network Model and Problem Statement}\label{FB-sec1}

This work focuses on the class of line networks. As illustrated in Fig.~\ref{Fig1}, $h$ denotes the number of hops in the network, and $V=\{v_0,v_1,\ldots,v_h\}$ and $\E=\{(v_i,v_{i+1}): i=0,\ldots,h-1\}$ to denote the set of nodes and the set of links in the network, respectively. Such a network has $h-1$ intermediate nodes, which are shown by black squares in the figure. Each intermediate node $v_i$ is assumed to have a buffer of $m_i$ packets. Note that buffer sizes of different nodes can be different. Without loss of generality, we assume $h\geq 2$ and $m_i>0$, for $i=1,\ldots,h-1$. Further, it is assumed that the destination node has no buffer constraints and that the source node possesses an infinitude of innovative packets at all times. The system is analyzed using a discrete-time model, where each node can transmit at most one packet over a link in any epoch. Intermediate buffers are assumed to be empty at epoch $l=0$ and the dynamics for $l\geq 0$ are steered by the loss processes on the edges of the network. The loss process on each link is assumed to be memoryless and statistically independent of the loss processes on other links. We let $\ve_{i+1}\in(0,1)$ to denote the erasure probability on the link $(v_i,v_{i+1})$ for $i=0,\ldots,h-1$. In this model, a node receives a packet on an incoming link when the neighboring upstream node transmits a packet and when the packet is not erased over the link. The reader is directed to Appendix~\ref{App.0-DiscvsCts} for a discussion on how the assumed discrete-time model relates to continuous-time exponential model that is commonly employed in queueing theory.

For the bulk of this work, we assume that the network has a perfect hop-by-hop feedback mechanism indicating the transmitting node of the receipt and storage of the transmitted packet by the receiving node. However, a subsequent section of this paper drops this assumption to study the capacity of line networks without feedback. It is also assumed in this work that nodes operate in a \emph{transmit-first} mode, i.e., each node first generates a packet (if it has a non-empty buffer) and transmits it on the outgoing edge. The node then processes the buffer after receiving the acknowledgement from the next-hop node before accepting/storing the packet on its incoming edge\footnote{Note that the need for such an ordering arises due to the discrete nature of time assumed in this work.}.

 For notational convenience, the random process on the link $(v_{i-1},v_i)$ is denoted by $\{X_i(l)\}_{\mbb{Z}_{\geq 0}}$. $X_i(l)=1$ if and only if the packet transmitted at epoch $l$ is deleted by the channel $(v_{i-1},v_{i})$, and $X_i(l)=0$ otherwise. For the sake of succinctness, we let $\mce\triangleq(\ve_1,\ldots,\ve_h)$ and buffer sizes $\mcm\triangleq(m_1,\ldots,m_{h-1})$.

\begin{figure}[h]
\psfrag{e0}{$\ve_1$} \psfrag{e1}{$\ve_2$} \psfrag{e2}{$\ve_h$} \psfrag{m1}{\hspace{-3mm}$m_1$} \psfrag{m2}{\hspace{-3mm}$m_2$} \psfrag{m3}{\hspace{-6mm}$m_{h-1}$}
\psfrag{v0}{$v_0$} \psfrag{v1}{$v_1$} \psfrag{v2}{$v_2$} \psfrag{v3}{$v_{h-1}$} \psfrag{v4}{$v_h$}
\centering
{\includegraphics[width=3.4in, height=1.0in,angle=0]{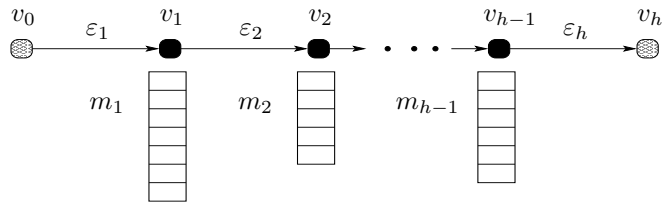}
\caption{An illustration of the line network.}\label{Fig1}}
\end{figure}

The focus of this paper is two-fold. The foremost aim is to identify the \emph{supremum} of all rates that are achievable by the use of any coding strategy between the ends of a line network with erasure probabilities $\mce$ and buffer sizes $\mcm$. In the line network illustrated in Fig.~\ref{Fig1}, we first aim to identify the maximum rate of information that the node $v_0$ can transmit to node $v_h$, which is denoted by $\mcc^{PF}(\mce,\mcm)$. The next issue on which we focus is the delay experienced by packets in intermediate node buffers when the network operates near the throughput capacity.

In our analysis, we employ the following notations. Vectors will be denoted by boldface letters, eg., $\mathbf{r},\mathbf{s}$. The indicator function for the set $\mathbbm{R}_{>0}=(0,\infty)$ is represented by $\sigma[\cdot]$. For any $x\in[0,1]$, $\ol{x}\triangleq 1-x$. The convolution operator is denoted by $\otimes$ and $\otimes^l f$ is used as a shorthand for the $l$-fold convolution of $f$ with itself. For $0<\lambda<1$, $\mg(\lambda)$ denotes the probability mass function of a positive random variable that is geometric with mean $\frac{1}{1-\lambda}$. For a discrete random variable $Z$ with probability mass function $f_Z$, $\langle Z \rangle$ and $\langle f_Z \rangle$ are both used to denote the mean of the random variable $Z$. Lastly, for appropriate $q\in\mathbbm{N}$, $\f_{q}$ denotes the Galois field of size $q$.

\section{Capacity of Line Networks}\label{FB-sec2}
In this section, we investigate the effect of finite buffers on the capacity of line networks. First, we present a framework for exact computation of the capacity of line networks that have perfect hop-by-hop feedback. We then present bounds on the capacity using techniques from queueing theory. Subsequently, we present our approaches to approximate the capacity of a line network. In the concluding subsection, we illustrate that the throughput capacity remains unaltered when feedback is absent provided packet-level network coding is allowed.
\subsection{Exact Computation of Capacity}\label{FB-Sec2.1}
The problem of identifying capacity is related to the problem of identifying schemes that are \emph{rate-optimal}. In the presence of lossless hop-by-hop feedback, the scheme performing the following steps in the given order is rate-optimal.
\begin{itemize}
\item[1.] If the buffer of a node is not empty at an epoch, then it must transmit one of the stored packets at that time.
\item[2.] A node deletes the packet transmitted at an epoch if it receives an acknowledgement of packet storage from the next-hop node at that epoch.
\item[3.] After performing $1$ and $2$, a node accepts an arriving packet if it has space in its buffer and sends an acknowledgment of packet storage to the previous node.
\end{itemize}
Notice that in the above scheme, at each epoch, the buffer of the last intermediate node is updated first, and the buffer of the first intermediate node is updated last. To determine the throughput capacity of the network, we need to track the number of packets that each node possesses at every instant of time by using the rules of buffer update under the above optimal scheme. Let $\mathbf{n}(l)=(n_1(l),\ldots,n_{h-1}(l))$ be the vector whose $i^{\textrm{th}}$ component denotes the number of packets $v_i$ possesses at time $l$. The variation of state at the $l^{\textrm{th}}$ epoch can be tracked using auxiliary random variables $Y_i(l)$ defined by
\begin{align}
Y_i(l)=\hspace{-1.5mm}\left\{\hspace{-1.5mm}\begin{array}{ll}\sigma[n_{i-1}(l)]X_i(l) & i=h\\
X_i(l)\sigma[n_{i-1}(l)(m_{i}-n_{i}(l)+Y_{i+1}(l))] & 1<i<h\\ X_i(l)\sigma[m_{i}-n_{i}(l)+Y_{i+1}(l)] & i=1
 \end{array}\right.\hspace{-2mm}.\label{FB-eqn2}
\end{align}

From the definition of the auxiliary binary random variables in (\ref{FB-eqn2}), we see that $Y_i(l)=1$ only if all the following three conditions are met:
\begin{itemize}
\item[1.] Node $v_{i-1}$ has a packet to transmit to $v_{i}$.
\item[2.] The link $(v_{i-1},v_{i})$ does not erase the packet at the $l^{\textrm{th}}$ epoch, {i.e.}, $X_{i}(l)=1$, and
\item[3.] Node $v_{i}$ is not full after its buffer update due to its transmission over $(v_{i},v_{i+1})$ at the $l^{\textrm{th}}$ epoch.
\end{itemize}
The changes in the buffer states can then by seen to be given by the following.
\begin{equation}
n_i(l+1)=n_i(l)+Y_i(l)-Y_{i+1}(l),\quad\quad 1\leq i < h.\label{FB-eqn3}
\end{equation}
Note that since $\mathbf{Y}(l)=(Y_1(l),\ldots,Y_h(l))$ is a function of $\bfn(l)$ and $\mathbf{X}(l)=(X_1(l),\ldots,X_h(l))$, $\bfn(l+1)$ depends only on its previous state $\bfn(l)$ and the channel conditions $\mathbf{X}(l)$ at the $l^{\textrm{th}}$ epoch. Hence, $\{\bfn(l)\}_{l\in\mathbbm{Z}_{\geq 0}}$ forms a Markov chain. The number of states corresponds to the number of possible assignments to $\mathbf{n}(l)$, which amounts to $\prod_{i=1}^{h-1} (m_i+1)$ possibilities. However, since at each time instant the number of packets that can be transmitted over any link is bounded by unity, we see that for every $i=1,\ldots,h-1\textrm{ and }  l\in\mathbbm{Z}_{\geq 0}$,
\begin{equation}
Y_i(l)\in\{0, 1\}\textrm{\quad and \quad} |n_i(l+1)-n_i(l)|\leq 1.\label{FB-eqn4}
\end{equation}
Therefore, the number of non-zero entries in each row of the probability transition matrix\footnote{The $ij^{\textrm{th}}$ term of the matrix $P(\mce,\mcm)$ represents the probability that the next state is $j$ given that state is presently $i$.} $P(\mce,\mcm)$ representing the transitions in the occupancy is bounded above by $\min(3^{h-1},\,\prod_{i=1}^{h-1}(m_i+1))$.

A detailed categorization of the states $\mathcal{S}$ that enables further understanding can be performed thus. We can order the states of the chain in such a way that the state $(s_1,\ldots,s_{h-1})\in\mathcal{S}$ corresponds to the row index $1+s_1+\sum_{i=2}^{h-1} s_i\prod_{j=1}^{i-1} (m_j+1)$ in the matrix $P(\mce,\mcm)$. Denote $T_\iota$ to be the set of states that have $s_{h-1}=\iota$ for $\iota=0,\ldots,m_{h-1}$. Let $\Gamma_\iota^-,\Omega_\iota,\Gamma_\iota^+$ represent the transition matrices for transitions from states in $T_\iota$ to those in $T_{\iota-1}$, $T_\iota$, $T_{\iota+1}$, respectively. Then, it can be shown that $\Gamma_i^+=\Gamma^+$, $\Omega_i=\Omega$, and $\Gamma_i^-=\Gamma^-$ for $\iota=1,\ldots,m_{h-1}-1$ (see Lemma~\ref{FB-lem1}). Therefore, the transition matrix of the chain can be structurally represented as follows.

\begin{equation}
\small{
\hspace{-0.5mm}{P}(\mce,\mcm)\hspace{-0.5mm}=\hspace{-1mm}\left(\hspace{-0.5mm}{}{ \begin{array}{lcccl}
{}{\Omega_0} \hspace{2.5mm} {}{\Gamma_0^+} \hspace{-4.5mm} &  {\mathbf{0}} \hspace{-4.5mm} & \cdots \hspace{-4.5mm} & \hspace{-4.5mm} & {\mathbf{0}} \\
{}\Gamma^-_1  \hspace{2.5mm} {}\Omega_1 \hspace{-4.5mm} & {}\Gamma^+_1 \hspace{-4.5mm} &  \cdots \hspace{-4.5mm} & \hspace{-4.5mm} &  {\mathbf{0}} \\ 
 {\mathbf{0}} \hspace{4.5mm} {}\Gamma^-_2 \hspace{-4.5mm} & {}\Omega_2 \hspace{-4.5mm} & \cdots \hspace{-4.5mm} & \hspace{-4.5mm} & {\mathbf{0}}  \\
 \hspace{-4.5mm} \hspace{2.5mm} & \vdots \hspace{-4.5mm} & \hspace{-4.5mm} & \hspace{-4.5mm} & \\
  {\mathbf{0}} \hspace{2.5mm} \hspace{-4.5mm} & \cdots \hspace{-4.5mm} & {}\Gamma^-_{m_{h-1}-1} \hspace{-4.5mm} & {}\Omega_{m_{h-1}-1} \hspace{-4.5mm} & {}\Gamma^+_{m_{h-1}-1} \\
 {\mathbf{0}} \hspace{2.5mm} \hspace{-4.5mm} & \cdots \hspace{-4.5mm} & {\mathbf{0}}\hspace{-4.5mm} & {}\Gamma_{m_{h-1}}^- \hspace{-4.5mm} & {}{\Omega_{m_{h-1}}}
 \end{array}}\hspace{-3.5mm}
  \right)\hspace{-1mm}. \nonumber}
\end{equation}

The dynamics given by the above equation can be depicted pictorially by the chain in Fig.~\ref{Fig2}. Note that due to the finite buffer condition and the non-negativity of occupancy, the transitions from the first block and from the last block differ from the transitions from the blocks between them. Further, the states within each $T_i$, $i=0,\ldots,m_{h-1}$, can be organized into $m_{h-2}+1$ sets in a similar fashion. In addition to this structural property, the transition sub-matrices satisfy the following algebraic properties.
\begin{figure}[ht!]
\psfrag{var1}{${\Omega_0}$}
\psfrag{var2}{$\vspace{6mm}{\Omega_1}$}
\psfrag{var3}{${\Omega_{m_{h-1}}}$}
\psfrag{M1}{${\Gamma_0^+}$}
\psfrag{M2}{$\vspace{4mm}{\Gamma^-_1}$}
\psfrag{M3}{${\Gamma^+_1}$}
\psfrag{M4}{$\vspace{4mm}{\Gamma^-_1}$}
\psfrag{M5}{$\vspace{2mm}\hspace{-2mm}{\Gamma_{m_{h-1}}^-}$}
\psfrag{S1}{$T_0$}
\psfrag{S2}{$T_1$}
\psfrag{S3}{$\hspace{-1mm}T_2$}
\psfrag{S4}{\hspace{-3.25mm}{$T_{m_{h-1}-1}$}}
\psfrag{S5}{\hspace{-3mm}$T_{m_{h-1}}$}
\centering
\includegraphics[width=3.4in,angle=0]{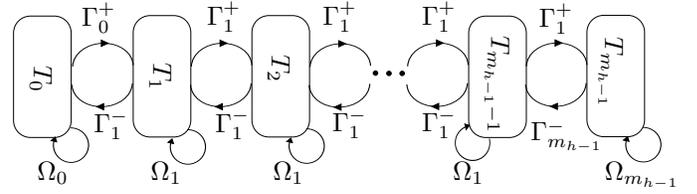}
\caption{The Markov chain for the dynamics of occupancy in a line network.} \label{Fig2}
\end{figure}

\begin{lem}\label{FB-lem1}
In a generic line network, the following hold.
\begin{itemize}
\item[a.] $\Gamma_i^+=\Gamma^+_1$, $\Omega_i=\Omega_1$, and $\Gamma_i^-=\Gamma^-_1$ for $i=1,\ldots,m_{h-1}-1$.
\item[b.] For $h\geq 2$, $\Gamma_i^-$ is non-singular and upper triangular for $i=1,\ldots, m_{h-1}$.
\item[c.] For $h>2$, $\Gamma_i^+$ is singular and lower triangular for $i=0,\ldots, m_{h-1}-1$ .
\item[d.] $I-\Omega_i$ is non-singular $i=0,\ldots, m_{h-1}$.
\end{itemize}
\end{lem}
\begin{proof}
See Appendix~\ref{App.1-1}.
\end{proof}

\psfrag{aa}{\footnotesize{$\varepsilon_1$}}
\psfrag{bb}{\footnotesize{$\overline{\varepsilon}_1$}}
\psfrag{cc}{\footnotesize{$\varepsilon_2$}}
\psfrag{dd}{\footnotesize{$\overline{\varepsilon}_2$}}
\psfrag{ee}{\footnotesize{$\varepsilon_3$}}
\psfrag{ff}{\footnotesize{$\overline{\varepsilon}_3$}}
\psfrag{bbcc}{\footnotesize{$\overline{\varepsilon}_1\varepsilon_2$}}
\psfrag{bbdd}{\footnotesize{$\overline{\varepsilon}_1\overline{\varepsilon}_2$}}
\psfrag{bbff}{\footnotesize{$\overline{\varepsilon}_1\overline{\varepsilon}_3$}}
\psfrag{bbee}{\footnotesize{$\overline{\varepsilon}_1\varepsilon_3$}}
\psfrag{aadd}{\footnotesize{$\varepsilon_1\overline{\varepsilon}_2$}}
\psfrag{aaee}{\footnotesize{$\varepsilon_1{\varepsilon}_3$}}
\psfrag{aaff}{\footnotesize{$\varepsilon_1\overline{\varepsilon}_3$}}
\psfrag{aacc}{\footnotesize{$\varepsilon_1\varepsilon_2$}}
\psfrag{ccff}{\footnotesize{$\varepsilon_2\overline{\varepsilon}_3$}}
\psfrag{aaddee}{\footnotesize{$\varepsilon_1\overline{\varepsilon}_2\varepsilon_3$}}
\psfrag{aaddff}{\footnotesize{$\varepsilon_1\overline{\varepsilon}_2\overline{\varepsilon}_3$}}
\psfrag{aaccff}{\footnotesize{$\varepsilon_1{\varepsilon}_2\overline{\varepsilon}_3$}}
\psfrag{bbccff}{\footnotesize{$\overline{\varepsilon}_1{\varepsilon}_2\overline{\varepsilon}_3$}}
\psfrag{bbccee}{\footnotesize{$\overline{\varepsilon}_1{\varepsilon}_2{\varepsilon}_3$}}
\psfrag{bbddee}{\footnotesize{$\overline{\varepsilon}_1\overline{\varepsilon}_2{\varepsilon}_3$}}
\psfrag{bbddff+aaee}{\footnotesize{$\overline{\varepsilon}_1\overline{\varepsilon}_2\overline{\varepsilon_3}+{\varepsilon}_1{\varepsilon}_3$}}
\psfrag{bbddff+aaccee}{\footnotesize{$\overline{\varepsilon}_1\overline{\varepsilon}_2\overline{\varepsilon_3}+{\varepsilon}_1{\varepsilon}_2{\varepsilon}_3$}}
\psfrag{ccee+bbdd}{\footnotesize{$\overline{\varepsilon}_1\overline{\varepsilon}_2+{\varepsilon}_1{\varepsilon}_3$}}
\psfrag{ccee+bbddff}{\footnotesize{$\overline{\varepsilon}_1\overline{\varepsilon}_2\overline{\varepsilon}_3+{\varepsilon}_2{\varepsilon}_3$}}
\psfrag{ee+bbddff}{\footnotesize{$\overline{\varepsilon}_1\overline{\varepsilon}_2\overline{\varepsilon}_3+{\varepsilon}_3$}}
\psfrag{bbddff+aaccee}{\footnotesize{$\overline{\varepsilon}_1\overline{\varepsilon}_2\overline{\varepsilon}_3+{\varepsilon}_1{\varepsilon}_2{\varepsilon}_3$}}
\psfrag{gg}{\small{$00$}}
\psfrag{hh}{\small{$01$}}
\psfrag{ii}{\small{$02$}}
\psfrag{jj}{\small{$12$}}
\psfrag{kk}{\small{$11$}}
\psfrag{ll}{\small{$10$}}
\psfrag{mm}{\small{$20$}}
\psfrag{nn}{\small{$21$}}
\psfrag{oo}{\small{$22$}}
\psfrag{T0}{\small{$T_0$}}
\psfrag{T1}{\small{$T_1$}}
\psfrag{T2}{\small{$T_2$}}
\begin{figure*}[ht!]
\centering
\includegraphics[height=3.5in,width=4.5in,angle=0]{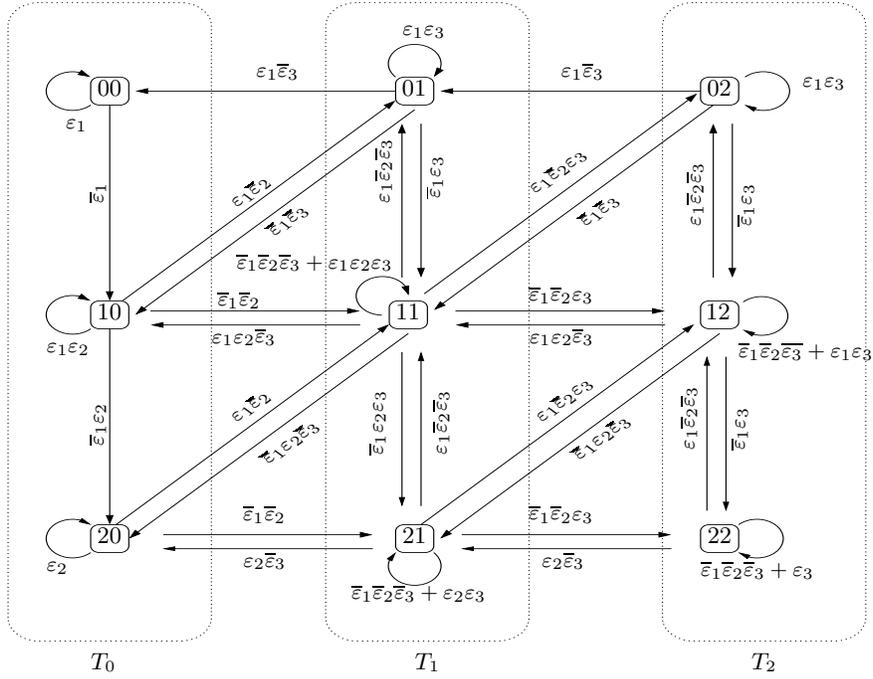}
\caption{Markov chain for a line network of three hops with erasure probabilities $\ve_1, \ve_2, \ve_3$ and intermediate nodes having a buffer size of two packets each.} \label{FB-MC}
\end{figure*}
To illustrate the implications of the above lemma, consider the Markov chain for a three-hop line network with erasure probabilities $\mce=(\ve_1,\ve_2,\ve_3)$, and with buffer sizes $\mcm=(2,2)$ presented in Fig.~\ref{FB-MC}. For this network, the algebraic properties of Lemma~\ref{FB-lem1} can be understood as follows.
 \begin{itemize}
 \item[1.] Any transition involving a decrease in the second component involves a non-negative change in the magnitude of the first component.
 \item[2.] Any horizontal transition involving a decrease in the second component is always feasible (provided the second component of the starting state is positive).
 \item[3.] Any transition involving an increase in the second component involves a non-positive change in the magnitude of the first component.
 \item[4.] Not all horizontal transitions involving an increase in the second component are feasible. For example, the  transitions from the state $(0,0)$ to $(0,1)$ and from the state $(0,1)$ to $(0,2)$ are infeasible, and hence  $(\Gamma_0^+)_{11}=(\Gamma_1^+)_{11}=0$.
 \end{itemize}

While the first two facts relate to the upper triangular structure and non-singularity of $\Gamma^-$, the latter two relate to the lower triangular and singularity properties of $\Gamma^+$.  This Markov chain for the dynamics of the state of a line network with perfect feedback is \emph{irreducible, aperiodic, positive-recurrent}, and \emph{ergodic}~\cite{MCBook, Feller1957}. By ergodicity, we can obtain temporal averages by statistical averages. Therefore, the throughput capacity $\mcc^{PF}(\mce,\mcm)$ can be identified by appropriately scaling the likelihood of the event that the system is in a state wherein the last node buffer is non-empty. This quantity is given by
 \begin{align}
 \mcc^{PF}(\mce,\mcm)&=\ol\ve_h\Pr[\{\bfs\in\mcs: s_{h-1}>0\}] \label{FB-eqn5.1}
 \end{align}
Notice that packets are not erased from the buffers without a receipt of acknowledgement of storage from the next-hop node. Therefore, the packet-flow rate is conserved. Therefore, for $0<i<h-1$, the throughput capacity can also be identified from
\begin{align}
 \mcc^{PF}(\mce,\mcm)&= \ol\ve_{i+1}\Pr\bigg[\Big\{\bfs\in\mcs: \begin{array}{l}s_{i}>0\\ s_{i+1}<m_{i+1}\end{array}\Big\}\bigg]. \label{FB-eqn5.1b}
 \end{align}
 Thus, the problem of identifying the capacity of line networks is reduced to the problem of computing the steady-state probabilities of the aforementioned Markov chain. However, due to the size of the Markov chain and its transition matrix, and the presence of multiple reflections due to the limited buffers at intermediate nodes, the problem of computing the steady-state distribution and capacity is computationally tedious even for networks of reasonable hop-lengths and buffer sizes.

 As the first step towards estimation, we can define a finite sequence of matrices by
\begin{equation}
H_i=\left\{\begin{array}{ll}I & i=0 \\ {\Gamma_1^-}^{-1}(I-\Omega_0) & i=1 \\
{\Gamma^-_{i}}^{-1}\big((I-\Omega_{i-1})H_{i-1} -\Gamma^+_{i-2}H_{i-2}\big) & 1<i \leq m_{h-1}
\end{array}\right..\nonumber
\end{equation}
Note that these matrices relate the steady-state distribution $\pi_{T_i}$ of the states in $T_i$, $i=0,\ldots,m_{h-1}$ by $ \pi_{T_i}=H_i\pi_{T_0}$. Using these relations, we can, in theory, estimate the capacity by
\begin{equation}
\mcc^{PF}(\mce,\mcm)\leq \ol\ve_h\Big(1-\frac{1}{||\sum_{j=0}^{m_{h-1}} H_j ||_1}\Big)
\end{equation}
 However, this matrix-norm approach does not provide insight into occupancy statistics of various nodes. Therefore, we focus on an approximations-based approach to capacity estimation in the remainder of this work.

\subsection{Bounds on the Capacity of Line Networks}\label{FB-Sec2.2}
In queueing theory, problems of identifying the steady-state probability of stochastic networks have often been dealt with approximations. Most approaches to problems in this area have been to approximate the dynamics of the network by focussing on local dynamics of the network around each node and the edges incident with it. The key idea in this section is to modify the exact Markov chain to derive bounds on throughput capacity. To do so, notice that the main reason for intractability of the exact system is the strong dependence of $Y_i(l)$ on not only $Y_{i-1}(l)$, but also $Y_{i+1}(l)$. This dependence translates to a strong dependence of $n_i(l)$ on both $n_{i-1}(l)$ and $n_{i+1}(l)$. Relaxation of this strong dependence will be a step towards possible decoupling of the system, and a deeper understanding of the tradeoffs in such networks.

Consider a network operation mode where each intermediate note transmits an acknowledgement whenever it \emph{receives} a packet (as opposed to the rate-optimal setting where it sends an acknowledgement whenever it \emph{receives and stores} a packet successfully). Under this new mode of operation, we notice that the dependence of the state of $i^{\textrm{th}}$ node on that of nodes further downstream is eliminated. This mode of operation is equivalent to assuming that a packet that arrives at a node whose buffer is full gets lost/dropped unlike the optimal mode of operation where it gets re-serviced. In this mode, the state updates are given by a simplified Markov chain that is generated by the following rule for all $l\in\mathbbm{Z}_{\geq 0}$ and $1\leq i<h$.
\begin{eqnarray}
\nt_i(l+1)= \nt_i(l)\hspace{-0.45mm}+\hspace{-0.45mm} \tilde{Y}_i(l)\sigma[m_i\hspace{-0.45mm}-\hspace{-0.45mm}\nt_i(l)\hspace{-0.45mm}+\hspace{-0.45mm}\tilde{Y}_{i+1}(l)]\hspace{-0.45mm}-\hspace{-0.45mm}\tilde{Y}_{i+1}(l), \hspace{-1.5mm}\label{AMCeqn}
\end{eqnarray}
where
\begin{eqnarray}
\tilde{Y}_i(l)=\left\{\begin{array}{ll}
\sigma[\nt_{i-1}(l)]X_i(l) & 1<i\leq h\\
X_i(l) & i=1\end{array}\right.. \label{FB-eqn6}
\end{eqnarray}

To avoid confusion, we appellate the chain that is obtained by the dynamics defined by (\ref{FB-eqn2}) and (\ref{FB-eqn3}) as the Exact Markov Chain (EMC) and the one defined by (\ref{AMCeqn}) and (\ref{FB-eqn6}) as the Approximate Markov Chain (AMC). Also, we allow $\bfn(l)$ and $\bfnt(l)$ to always denote the state of an instance of the process generated by the EMC and the AMC, respectively. Then, the following property holds.

\begin{thm}\emph{(Temporal Boundedness Property of the AMC)}\label{FB-thm1}
Consider a line network with $h$ hops and an instance of channel realizations $\{X_i(l):i=1,\ldots, h\}_{l\in\mathbbm{Z}_{\geq 0}}$. Suppose we track the variation of the states of the EMC and the AMC using this instance of channel realizations with the same initial state $\bfn(0)=\bfnt(0)$. Then, for any $l\in\mathbbm{Z}_{\geq 0}$ and $1\leq i\leq h-1$, the following holds.
\begin{equation}
n_i(l)\geq \tilde{n}_i(l).
\end{equation}
\end{thm}
\begin{proof}
The proof is detailed in Appendix~\ref{App.1-2}.
\end{proof}

 The Temporal Boundedness Property guarantees that statistically, the probability that a node has an empty buffer is overestimated by the AMC. In fact, if we can identify the steady-state distribution of the states of AMC,  we can provide a lower bound for the steady-state probability of any subset of states $\mathcal{A}\subseteq \mcs$ that have the form
\begin{equation}
\mathcal{A}=\Big\{\bfs\in\mcs: (s_j\geq a_j),\, j=1,\ldots h-1\Big\}, \label{FB-LBsets}
\end{equation}
where $0\leq a_j\leq m_j$ for $j=1,\ldots,h-1$. Using the Temporal Boundedness property in conjunction with (\ref{FB-eqn5.1}), we can provide a lower bound $\ul\mcc^{PF}(\mce,\mcm)$ for the capacity of the line network by underestimating the probability in (\ref{FB-eqn5.1}) by using the steady-state distribution of the AMC instead of that of the EMC. Equivalently, the capacity of the line network is at least that of the throughput achievable by the AMC. This above idea of lower bound extends easily to an upper bound using the following result. The fundamental idea behind the following bound is to manipulate the buffer sizes at each node so that the packet drop in the modified network is provably infrequent than in the actual network.
\begin{thm}\label{FB-Ubthm}
   Let the operator $\mathtt{\Sigma}$ be defined by $(b_1,\ldots,b_k)\stackrel{\mathtt{\Sigma}}{\mapsto}(b_1,b_1+b_2,\ldots,b_1+\ldots+b_k)$. For a given network with distinct erasure probabilities $\mce$ and buffer sizes $\mcm$, denote $\ol\mcc^{PF}(\mce,\mcm)$ to be the throughput computed from the steady-state distribution of the AMC defined by (\ref{AMCeqn}) and (\ref{FB-eqn6}) with erasure probabilities $\mce$ and buffer sizes $\mathtt{\Sigma}(\mcm)$, i.e., $\ol\mcc^{PF}(\mce,\mcm)\triangleq \ul\mcc^{PF}(\mce,\mathtt{\Sigma}(\mcm))$. Then,
\begin{equation}
\mcc^{PF}(\mce,\mcm) \leq \ol\mcc^{PF}(\mce,\mcm)
\end{equation}
\end{thm}
\begin{proof}
A detailed proof is presented in Appendix~\ref{App.1-4}.
\end{proof}

Thus, the problem of bounding capacity is reduced to identifying the steady-state probability of the AMC. Notice that the above bounds are not in a computable form, since they still involve identifying the steady-state distribution of the AMC. Even though the AMC is significantly simpler than the EMC, the output process from each intermediate node is not renewal~\cite{Serfozo_QT0}. Therefore, the distribution of inter-departure times from each intermediate node is insufficient to completely describe the arrival process at intermediate nodes $v_i$ for $1<i<h$. Therefore, a straightforward hop-by-hop analysis (without further assumptions) seems insufficient to identify the capacity of such networks.

\subsection{Iterative Estimation of the Capacity of Line Networks}\label{FB-Sec2.3}

In this section, we present two iterative estimates for the capacity of line networks that is based on certain simplifying assumptions regarding the EMC. We notice that the difficulty of exactly identifying the steady-state probabilities of the EMC stems from the finite buffer condition that is assumed. The finite buffer condition introduces a strong dependency of state update at a node on the state of the node that is downstream. This effect is caused by blocking when the state of a node is forced to remain unchanged because the packet that it transmitted is successfully delivered to the next-hop node, but the latter is unable to store the packet due to lack of space in its buffer. Additionally, the non-tractability of the EMC is compounded by a non-renewal packet departure process from each intermediate node. In this section, we ignore some of these issues to develop iterative methods for estimation. Figure~\ref{FB-IE} encapsulates the assumptions made in both estimation approaches. While both approaches ignore the non-renewal nature of packet arrival process at each node, the first approach makes an additional memoryless assumption on the arrival process. Additionally, both approaches model the effect of blocking by the introduction of a single parameter $p_b$ that represents the probability that an arriving innovative packet will be blocked.

\begin{figure*}[ht!]
\psfrag{Sta_i}{\footnotesize{$n_{i-1}{(l)}$}}
\psfrag{Sta_k}{\hspace{1mm}\footnotesize{{$n_i(l)$}}}
\psfrag{Blocking}{\hspace{2mm}\scriptsize{Memoryless Blocking}}
\psfrag{Pb}{\footnotesize{${P_b}_i$}}
\psfrag{Xi(l)}{\hspace{0mm}\scriptsize{Renewal process}}
\psfrag{Yi(l)}{\footnotesize{$\{Y_i(l)\}$}}
\centering
\includegraphics[height=1.75in, width=4.75in,angle=0]{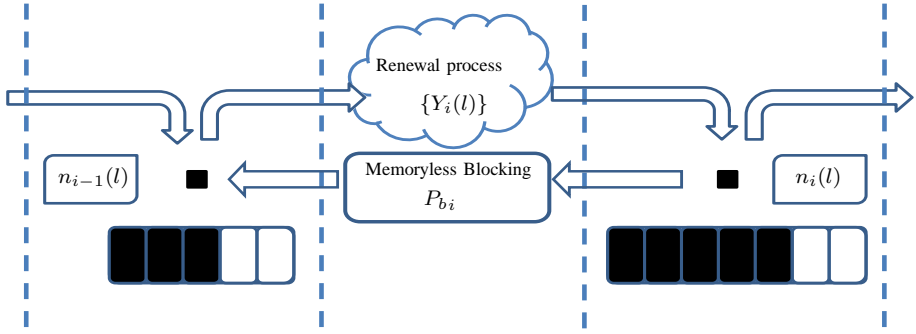}
\caption{Illustration of the assumptions in iterative estimation.} \label{FB-IE}
\end{figure*}

\subsubsection{Rate-based Iterative Estimate}\label{FB-Sec2.3.1}
This estimate makes the following assumptions to decouple the dynamics of the system and enable capacity estimation.
\begin{itemize}
\item[A1.] The packet departure process at each intermediate node is memoryless. In other words, each node $v_i$ sees a packet arrival process that is memoryless with (average) rate $r_i$ packets/epoch. This assumption allows us to track information rates over links while simplifying the higher order statistics.
\item[A2.] Any packet that is transmitted unerased by the channel $(v_i,v_{i+1})$ is blocked independently with a probability ${p_b}_{i+1}$. That is, for any $0< k\leq m_i$,
\begin{align}
\Pr[Y_{i+1}(\cdot)=0,X_{i+1}(\cdot)=1|n_i(\cdot)=k] = \ol\ve_{i+1}{p_b}_{i+1}.\nonumber
\end{align}

Here, ${p_b}_{i+1}$ denotes the blocking probability due to full buffer state at $v_{i+1}$. This assumption allows us to track the blocking probability ignoring higher order statistics of the blocking process.
\item[A3.] For each node $v_i$ and epoch $l$, the event of packet arrival and the event of blocking from $v_{i+1}$ are independent of each other.
\end{itemize}
\begin{figure}[ht!]
\psfrag{aa}{$\footnotesize{\ol\alpha_0}$}
\psfrag{bb}{$\footnotesize{\hspace{-2mm}\ol{\alpha+\beta}}$}
\psfrag{cc}{$\footnotesize{\hspace{1mm}\ol\beta}$}
\psfrag{dd}{$\footnotesize{\alpha_0}$}
\psfrag{ee}{$\footnotesize{\alpha}$}
\psfrag{ff}{$\footnotesize{\beta}$}
\psfrag{gg}{$\footnotesize{\beta}$}
\centering
\includegraphics[width=3.4in,angle=0]{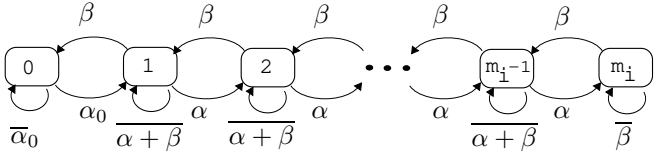}
\caption{The Markov chain for $n_i$ under the simplifying assumptions.}\label{FB-1HMC}
\end{figure}
The above assumptions are valid in the limiting case of large buffers provided the system corresponds to a stable queueing configuration. By assuming that they hold in general, the effect of blocking is spread equally over all non-zero states of occupancy at each node. Similarly, the assumptions also spread the arrival rate equally among all occupancy states. Given that the arrival rate of packets at the node $v_i$ is $r_i$ packets/epoch, and the blocking probability of the next node is ${p_b}_{i+1}$, the local dynamics of the state change for the node $v_i$ under assumptions A1-A3 is given by the Markov chain of Fig.~\ref{FB-1HMC} with the parameters set to the following.
\begin{align}
\begin{array}{ccl}
\alpha&=&r_i(\ve_{i+1}+\ol\ve_{i+1}{p_b}_{i+1})\\
\beta&=&(1-r_i)\ol{p_b}_{i+1}\ol\ve_{i+1}\\
\alpha_0&=&r_i
\end{array}.\label{Eqn-24}
\end{align}
Using these parameters, the steady-state distribution\footnote{If ${p_b}_{i+1}=1$, then we set $\vp(m_i| r_i,\ve_{i+1},{p_b}_{i+1})=1$.} $\{\phi_{v_i}(k)\triangleq\vp(k| r_i,\ve_{i+1},{p_b}_{i+1})\}_{k=0}^{m_i}$ of the chain of Fig.~\ref{FB-1HMC} can be computed to be
\begin{align}
\hspace{-3mm}\vp(k| r_i,\ve_{i+1},{p_b}_{i+1})\triangleq\hspace{-1mm}\left\{\begin{array}{l}\hspace{-2mm}\frac{1}{1+\frac{\alpha_0}{\beta}\big(\sum_{l=0}^{m_i-1}\frac{\alpha^l}{\beta^l}\big)} \hspace{3mm}k=0\\
\hspace{-2mm}\frac{\frac{\alpha_0\alpha^{k-1}}{\beta^k}}{1+\frac{\alpha_0}{\beta}\big(\sum_{l=0}^{m_i-1}\frac{\alpha^l}{\beta^l}\big)} \hspace{3mm} 0<k\leq m_i\end{array}\hspace{-3mm}\right.. \label{FB-SS}
\end{align}
 Assuming that $v_i$ observes a packet arrival rate of $r_i$ from $v_{i-1}$ and a blocking probability of ${p_b}_{i+1}$ from $v_{i+1}$, the blocking probability ${p_b}_i$ that the node $v_{i-1}$ perceives from the node $v_i$ and the arrival rate $r_{i+1}$ that $v_{i+1}$ observes can be computed via (\ref{FB-SS}) using the following equations.
\begin{align}
{p_b}_i&=(\ve_{i+1}+\ol\ve_{i+1}{p_b}_{i+1})\vp(m_i| r_i,\ve_{i+1},{p_b}_{i+1})\label{FB-PB}\\
r_{i+1}&=\ol\ve_{i+1} (1-\vp(0| r_i,\ve_{i+1},{p_b}_{i+1})) \label{FB-R}
\end{align}
Note that the blocking probability ${p_b}_i$ is computed using the full occupancy probability of the node $v_i$. While in reality, a packet is blocked by $v_i$ only if at the arriving instant, the node has full occupancy, A2 models any arriving packet to be blocked with the above probability irrespective of the occupancy of $v_i$. Also, in (\ref{FB-PB}) and (\ref{FB-R}) the arrival rate from the node $v_1$ is $r_1=\ol\ve_1$ and the blocking probability ${p_b}_h=0$.

Given two vectors $\mcr=(r_1,\ldots,r_h)\in[0,1]^{h}$ and  $\mcp=({p_b}_1,\ldots,{p_b}_h)\in[0,1]^{h}$, we term $(\mcr,\mcp)$ as a rate-approximate solution to EMC, if they satisfy the equations (\ref{FB-PB}) and (\ref{FB-R}) in addition to
having $r_1=\ol\ve_1$ and ${p_b}_h=0$. Since these relations were obtained from making assumptions on the EMC, it is \emph{a priori} unclear if there exist rate-approximate solutions for a given system $(\mce, \mcm)$. Fortunately, the following result guarantees both the uniqueness and an algorithm for identifying the rate-approximate solution to the EMC.
\begin{thm}\label{FB-uniqthm}
Given a line network with link erasures $\mce=(\ve_1,\ldots,\ve_h)$ and intermediate node buffer sizes $\mcm=(m_1,\ldots, m_{h-1})$, there is exactly one rate-approximate solution $(\mcr^*(\mce,\mcm),\mcp^*(\mce,\mcm))$ to the EMC. Further, the rate-approximate solution satisfies flow conservation. That is,
\begin{equation*}
r^*_i(1-{{p_b}}^*_i)=r^*_j(1-{{p_b}}^*_j), \quad \quad 1\leq i,j\leq h.
\end{equation*}
\end{thm}
\begin{proof}
The proof is detailed in Appendix~\ref{App.1-3}
\end{proof}
Finally, the estimate of the capacity can be obtained from the rate-approximate solution by computing the average rate of packet storage at each node using
\begin{equation}
\mcc^*(\mce,\mcm)=r^*_i(1-{p_b}^*_i), \quad\quad i=1,\ldots,h. \label{FB-LinCapEst}
\end{equation}
Note that by the conservation of flow, any $i\in\{1,\ldots,h\}$ can be used in the above equation to identify capacity.

As an illustration, consider a simple four-hop network with erasures $\mce=(0.5, 0.4999, 0.4998, 0.4)$ and buffer sizes $\mcm=(5,5,5)$.  From the above estimation method, we arrive at
 \begin{eqnarray}
\mcr^*&=&(0.5, 0.46797, 0.43958, 0.43484),\\
\mcp^*&=&(0.13031, 0.07078, 0.01076, 0),\\
\mcc^*(\mce,\mcm)&=&0.43484 \textrm{ packets/epoch}.
  \end{eqnarray}
 From simulations, the throughput capacity was found to be $0.43501$ packets/epoch for the same network.

\subsubsection{Distribution-based Iterative Estimate}\label{FB-Sec2.3.2}
In this section, we assume that the given line network $(\mce,\mcm)$ satisfies $\ve_i\neq\ve_j$ for $i\neq j$. Since the capacity of a line network is a continuous function of the system parameters, this assumption is not restrictive. A system with non-distinct erasure parameters can be approximated to any degree of precision by a system with distinct erasure probabilities.

Before we introduce the second approach for estimation, we present the following technical result\footnote{For this theorem, note that we do not require that all $p_i$s or all $p_i'$s be positive. We only need that their sum be unity and that they generate a valid probability distribution, respectively.} wherein we denote $\mathbbm{I}$ to be the identity distribution for the convolution operator.
\begin{thm}\label{FB-LBthm}
Consider a tandem queueing system of two nodes where the first node possessing $m$ buffer slots is fed by a renewal process whose inter-arrival time distribution is $g^{\textrm{in}}=\sum_{i=1}^{N-1} p_i\mathbbm{G}(\tee_i)$ with $p_1,\ldots,p_{N-1}\in\mathbbm{R}$ and $\tee_{i}\neq \tee_j$ for $1\leq i<j\leq N-1$. Suppose that the distribution of service time is $\mathbbm{G}(\tee_{N})$, where $\tee_{N}\neq \tee_i$ for $1\leq i\leq N-1$. Further, suppose that the second node blocks an arriving packet memorylessly with probability $q\in(0,1)$, and that any blocked packet gets re-serviced. Then, the distribution of inter-arrival times as seen by the second node is given by
\begin{equation}
 g^{\textrm{out}}=\Upsilon(g^{\textrm{in}},m,\tee_N,q)\otimes G(\tee_N),
 \end{equation}
where $\Upsilon(g^{\textrm{in}},m,\tee_N,q) = \Big(\ol\alpha\mathbbm{I}+\alpha\sum_{l=1}^{N-1}p'_l\mathbbm{G}(\tee_l)\Big)$ for some $0<\alpha<1$ and $p'_l\in\mathbbm{R}$, $l=1,\ldots,N-1$, with $\sum_l p'_l=1$.
\end{thm}
\begin{proof}
A detailed analysis including the means of identifying $\alpha,\{p'_i:i=1,\ldots, N\}$ is given in Appendix~\ref{App.2}.
\end{proof}

Just as in the Rate-based Iterative Estimate, this estimate also makes three assumptions to simplify the EMC. While the Distribution-based Iterative Estimate makes assumptions A2 and A3, it relaxes assumption A1 to the following:
\begin{itemize}
\item[A1$^*$.] The packet departure process at each intermediate node is renewal.
\end{itemize}
Note that Assumption A1 allows for tracking only the average rate of information flow on edges whereas A1$^*$ allows tracking of the distribution of packet inter-arrival times. However, A1$^*$ ignores the fact that the distribution of an inter-arrival time changes with the knowledge of past inter-arrival times. To track the inter-arrival distribution and blocking probabilities at each node, the Distribution-based Iterative Estimate uses Theorem~\ref{FB-LBthm} in a hop-by-hop fashion. Assuming that the packet arrival process at $v_i$ is renewal with an inter-arrival distribution $f_{i}$, and that the memoryless blocking from $v_{i+1}$ occurs with probability ${p_b}_{i+1}$, we see that the packet inter-arrival distribution seen by $v_{i+1}$ is  given by
\begin{eqnarray}
f_{i+1}=\Upsilon(f_i,m_i,\ve_{i+1},{p_b}_{i+1})\otimes\mathbbm{G}(\ve_{i+1}). \label{FB-R-DIE}
\end{eqnarray}
Notice that just like in (\ref{Eqn-24}), $\Upsilon$ uses the effective erasure probability to incorporate the effect of blocking by $v_{i+1}$. However, this corrective term does not appear in $\mathbb{G}(\cdot)$ term, because $f_{i+1}$ represents the distribution of packet inter-arrival times at $v_{i+1}$, and not the distribution of the time between two adjacent successful packet storages at $v_{i+1}$. Further, the blocking probability of $v_i$ as perceived by $v_{i-1}$ is given by
\begin{align}
{p_b}_{i}&=\Pr[\textrm{A packet arriving at $v_i$ sees full buffer}]\\
&\stackrel{(\ref{FB-BlockingEqn})}{=} \mathcal{P}(f_i,m_i,\ve_{i+1},{p_b}_{i+1}). \label{FB-Pb-DIE}
\end{align}
Just as in the Rate-based Iterative Estimate, we call a solution to (\ref{FB-R-DIE}) and (\ref{FB-Pb-DIE}) with boundary conditions ${p_b}_h=0$ and $f_1=\mathbbm{G}(\ve_1)$ as a distribution-approximate solution. Though the existence and uniqueness of the distribution-approximate solution for a given system $(\mce,\mcm)$ has eluded us, simulations reveal that for each system, the solution is unique and can be found by iteratively using the following algorithm.

\begin{algorithm}[h!]
\small{\caption{\emph{Distribution-based Iterative Estimate}}\label{alg:DbIE}
\begin{algorithmic}[1]
\STATE $\texttt{Count}=1$ and ${p_b}_i[\texttt{Count}]=0$, $i=1,\ldots,h-1$.
\WHILE {\texttt{Count}$\leq$\texttt{Max\_Iter}}
\STATE $f_1[\texttt{Count}]=\mathbbm{G}(\ve_1)$, ${p_b}_h[\texttt{Count}]=0$, and $j=1$.
\WHILE {$j<h$}
\STATE Compute $f_{j+1}[\texttt{Count}]$, ${p_b}_j[\texttt{Count}+1]$ employing (\ref{FB-R-DIE}) and (\ref{FB-Pb-DIE}) (that use $f_{j}[\texttt{Count}]$, ${p_b}_{j+1}[\texttt{Count}]$)
 \STATE $j\leftarrow j+1$.
\ENDWHILE
\STATE $\texttt{Count}\leftarrow\texttt{Count}+1$.
\ENDWHILE
\end{algorithmic}}
\end{algorithm}

Note that during any round of $\texttt{Count}$ in the above algorithm, (\ref{FB-R-DIE}) can be iteratively used to identify $f_i[\texttt{Count}]$ in Step 5 only if the output distribution of inter-departure times from each node is a weighted sum of geometric distributions. This is however guaranteed if the erasure probabilities of no two links are equal. Alternately, Step 2 can be replaced by a convergence-type criterion instead of the \texttt{Max\_Iter} criterion. After sufficiently large number of iterations, the distributions and blocking probabilities usually converge (to $f_i^\star$ and ${p_b}_{i}^\star$), and  upon convergence the capacity can be estimated via
\begin{equation}
\mcc^\star(\mce,\mcm)=\frac{1}{\langle f_h^\star\rangle}.
\end{equation}
Using the above approach for the four-hop example network at the end of Sec.~\ref{FB-Sec2.3.1}, we have
\begin{align}
f_4^\star&=138240.92\mathbbm{G}(0.5)-275765.59\mathbbm{G}(0.4999)\nonumber\\
&\hspace{5mm}+137525.64\mathbbm{G}(0.4998)+0.03\mathbbm{G}(0.4),\\
\mathbf{p_b}^\star&=( 0.12983, 0.070006, 0.010406),\\
\mathcal{C}^\star(\mce,\mcm)&= 0.435089 \textrm{ packets/epoch}.
\end{align}

\subsection{Capacity of Line Networks without Feedback}\label{FB-Sec2.4}
Feedback from next-hop node provides a natural means of buffer update and packet deletion. In the absence of feedback, due to the finiteness of buffers, each intermediate node must have a local rule for buffer update to accept packets that arrive. A rule for packet deletion or update must be maintained at each node so that the buffers are used efficiently. A network coded-scheme based on random linear combinations over a large finite field $\fq$ of size $q$ as is described in~\cite{TracyHoJrnl} presents an effective means of buffer update and packet delivery. Consider the following scheme based on network coding.
\begin{itemize}
\item[1.] At each epoch, a node having a buffer size of $m$ packets picks a vector $\mathbf{a}\in\fq^m$ uniformly at random to generate a random linear combination in the following manner. For each buffer slot $i\in\{1,\ldots,m\}$, the packet $P_i$ stored in that slot is represented as a vector over $\fq$ and the output packet is generated by computing $\sum_{i=1}^m a_i P_i$. This generated packet is then transmitted during the epoch.
\item[2.] If a packet $P$ is received by a node at an epoch, it first generates the output packet at that instant and then updates its buffer in the following manner. The node selects a vector $\mathbf{b}\in\fq^m$ uniformly at random and for each $k\in\{1,\ldots,m\}$, adds the packet $b_k P$ to the packet stored in the $k^\textrm{th}$ buffer slot.
\end{itemize}

Note that in the network coding scheme described above, after sufficient time after the commencement of packet transfer from the source, all buffer slots of every intermediate node almost always have non-trivial contents unlike the scheme with perfect feedback. However, it is not true that all of these packets are \emph{innovative},  {i.e.}, packets may contain common \emph{information}\footnote{Here, we use information to represent the number of linearly independent packets w.r.t. the chosen base field. A set $S=\{P_1,\ldots,P_N\}$ is said to contain $n$ packets of information if $\dim(\lsp(S))=n$.}. Such a condition may occur when the packets are linearly dependent in the algebraic sense. With this notion of information, the rate of information received by the destination node can be seen to be the asymptotic rate of arrival of innovative packets. The following result characterizes these rates achieved by the network coding scheme over the field $\fq$ and relates it to throughput capacity in the presence of lossless feedback $\mcc^{PF}(\mce,\mcm)$.

\begin{thm}\label{Thm-NFOpt}
Let $\mcc^{NF}_{\f_{q}}(\mce,\mcm)$ denote the rate of arrival of innovative packets at the destination node of a line network without feedback assuming that the aforementioned network coding scheme over the field $\f_{q}$ is employed. Then, for each sequence of finite fields $\{\f_{q_l}\}_{l\in\mathbbm{N}}$ such that ${q_l}\ra\infty$, we have
\begin{equation}
\lim_{l\ra\infty}\mcc^{NF}_{\f_{q_l}}(\mce,\mcm) = \mcc^{PF}(\mce,\mcm). \label{FB-eqn1}
\end{equation}
\end{thm}
\begin{proof}
The proof is presented in Appendix~\ref{App.3}.
\end{proof}
From (\ref{FB-eqn1}), we observe that there is no loss in achievable rates when feedback is absent and that the aforementioned network coding scheme is rate-optimal for line networks without feedback, provided a large field size is employed.

\section{Packet Delay Distribution}
\label{FB-Delay}
In this section, we use the iterative estimates of Section~\ref{FB-Sec2.3} to obtain estimates on the probability distribution of the delay experienced by information packets in line networks with perfect feedback under the optimal strategy of Section~\ref{FB-Sec2.1}. We abstain from defining latency of data packets in networks without feedback, since optimal schemes for such networks involve packet-level coding.

When perfect feedback is available, we define the delay of a packet as the time taken from the instant when the packet is stored in the buffer of the first intermediate node to the instant when the destination receives it. Since the delay statistics depend on how the packets are handled in intermediate nodes, in addition to the optimal scheme of Section~\ref{FB-Sec2.1}, we assume a \emph{first-come first-serve} treatment of packets at intermediate node buffers. Note that this assumption is made only for the ease of presentation. The framework permits the analysis of randomized schemes where each node after a successful transmission selects a packet in its buffer randomly and memorylessly, and transmits it repeatedly until it is stored at the next-hop node.

In order to compute the distribution of delay that a packet experiences in the network, one can proceed in a hop-by-hop fashion using two parameters: (1) $\rho_j$, an estimate of blocking probability at node $v_i$, $i=1,\ldots,h$, and (2) $\psi_i(k)$, an estimate of the distribution of occupancy at node $v_i$ (for each $i=1,\ldots,h-1$ and $k=0,\ldots,m_i$) just before packet arrival conditioned on the event that the arriving packet is successfully stored.

In the last relay node $v_{h-1}$, the additional delay perceived by a packet arriving at $l^\textrm{th}$ epoch depends on the occupancy $n_{h-1}(l)$ of the node $v_{h-1}$ and $\ve_h$. Suppose at epoch $l$, node $v_{h-1}$ has $k\leq m_{h-1}-1$ packets excluding the arriving packet. Then, the packet has to wait for the $k$ already-stored packets to leave before it can be serviced. Since the services are memoryless, the distribution of delay is given by a sum of $k+1$ independent geometric distributions each with a mean inter-arrival time $\frac{1}{1-\ve_h}$, i.e., ${\otimes^{k+1}\mathbbm{G}({\ve_h})}$. Hence, the distribution of additional delay induced by waiting in the buffer of $v_{h-1}$ is
\begin{equation}
{\mc{D}_{h-1}=\sum_{i=0}^{m_{h-1}-1} {\psi_{h-1} (i) \big[{\otimes^{i+1}\mathbbm{G}({\ve_h})}}\big].     \label{last_relay_delay}}
\end{equation}
However, the situation is different for other intermediate delays because of the effect of blocking. The additional delay incurred while being stored in the node $v_j,\, 0<j<h-1$, is given by
\begin{equation}
\mc{D}_{j}=\sum_{i=0}^{m_{j}-1} {\psi_{j} (i) \big[{\otimes^{i+1}\mathbbm{G}({\ve_{j+1}+{\rho}_{j+1}\ol\ve_{j+1}})}\big]}, \label{int_relay_delay}
\end{equation}
since a packet is deleted from the buffer of $v_j$ only if the channel successfully transmits it and $v_{j+1}$ does not block the arriving packet, which by assumption A2 occurs memorylessly with a probability $\ve'_{j+1}\triangleq\ve_{j+1}+{\rho}_{j+1}\ol\ve_{j+1}$. Assuming that the delays incurred by waiting in the buffer of each node is independent of each other, we obtain the total delay considering all hops to be
\begin{equation}
{\mc{D}=\mc{D}_{1}\otimes \cdots \otimes \mc{D}_{h-1}.}
\label{Delay_dist_formula}
\end{equation}
Note that in addition to the above delay, the source node attempts to transmit the packet multiple times before the packet is successfully accepted at the first intermediate node. The distribution of this time spent in this is approximately given by a random variable whose distribution is $\mathbbm{G}(\ve_{1}')$. Thus, the iterative estimation technique provides us with a framework to approximately but analytically compute the delay profile using an estimate of distribution of packets seen by an arriving packet that is successfully stored, and an estimate of the blocking probabilities.

Finally, the pair of estimates ($\psi_j (\cdot)$, $\rho_j$) can be obtained from the rate-approximate solution by using
\begin{align}
\rho_j&={p_b}_i^*,\\
\hspace{-2mm}\psi_j (i)&=\left\{\begin{array}{ll}
\frac{\phi_{v_j}^*(i)+\phi_{v_j}^*(i+1)(1-\ve_{j+1}')}{1-\phi_{v_j}^*(m_j)\ve_{j+1}'} & i=0\\ \frac{\phi_{v_j}^*(i)\ve_{j+1}'+\phi_{v_j}^*(i+1)(1-\ve_{j+1}')}{1-\phi_{v_j}^*(m_j)\ve_{j+1}'} & 1\leq i <m_j
\end{array}\hspace{-2mm}\right..
\end{align}

Similarly, another pair of estimates can be obtained using the Distribution-based Iterative Estimate by
\begin{align}
\rho_j&={p_b}_i^\star,\\
\psi_j (i)&=\left\{\begin{array}{ll}
\frac{\pi_j^\star(i+1)}{1-{p_b}_j} & i=0,\ldots,m_j-2\\
\frac{\pi_j^\star(m)-{p_b}_j^\star}{1-{p_b}_j^\star} & i=m_j-1
\end{array}\right.,
\end{align}
where $\pi_j^\star(\cdot)$ is the eigenvector of the (\ref{App2-eqn1}) upon convergence. Combining the above equations, two estimates for the delay profile for line networks with feedback can be obtained.

\section{Results of Simulation}\label{FB-sec4}

In this section we present the results of simulation comparing our analytic results to simulations of line networks with perfect feedback. First, the simulations for the capacity are presented, and then the simulations for delay profiles are presented. This section ends with a discussion on the efficient usage of buffers and the interplay of buffer size, capacity and delay.

In our model, a line network is completely defined by the number of hops, the erasure probability for each link and the buffer size at each intermediate node. To study the accuracy of our bounds and estimates, we vary one of these three parameters while keeping the remaining two fixed. In each of the figures, the actual capacity and bounds obtained via simulations are presented in addition to our estimates. Further, for the sake of brevity, we abbreviate Distribution-based Iterative Estimate (Algorithm~\ref{alg:DbIE})), Rate-based Iterative Estimates (Algorithm~\ref{alg:RbIE}), Lower Bound (Thm.~\ref{FB-thm1}), and Upper Bound (Thm.~\ref{FB-Ubthm}) to DbIE, RbIE, LB, and UB, respectively.

Figure~\ref{FB-LineL} presents the variation of the capacity
\begin{figure}[ht!]
\centering
\psfrag{Xaxis}{\small{ Number of hops $h$}}
\psfrag{Yaxis}{Capacity (packets/epoch)}
\psfrag{dddddddddddddddata1}{\scriptsize{DbIE ($\ve=0.25$)}}
\psfrag{data2}{\scriptsize{RbIE. ($\ve=0.25$)}}
\psfrag{data3}{\scriptsize{Capacity ($\ve=0.25$)}}
\psfrag{data5}{\scriptsize{UB ($\ve=0.25$)}}
\psfrag{data4}{\scriptsize{LB ($\ve=0.25$)}}
\psfrag{data6}{\scriptsize{DbIE ($\ve=0.50$)}}
\psfrag{data7}{\scriptsize{RbIE ($\ve=0.50$)}}
\psfrag{data8}{\scriptsize{Capacity ($\ve=0.50$)}}
\psfrag{data10}{\scriptsize{UB ($\ve=0.50$)}}
\psfrag{data9}{\scriptsize{LB ($\ve=0.50$)}}
\includegraphics[height=2.25in,width=3.4in,angle=0]{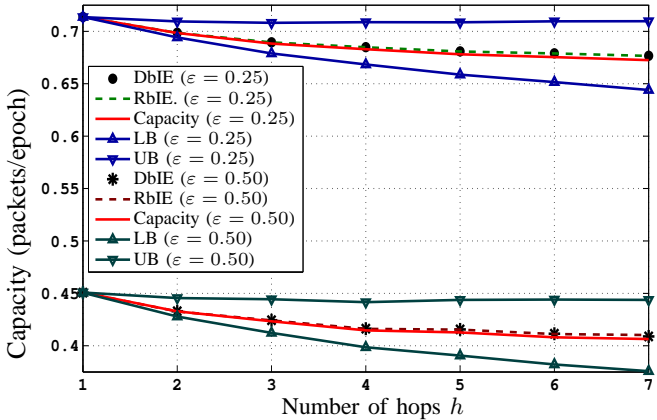}
\caption{Capacity of line networks with $m=5$ and varying hop-length $h$.}\label{FB-LineL}
\end{figure}
with the number of hops of line networks when each intermediate node possesses a buffer size of five packets. The figure presents simulations for networks when the probability of erasure on each link is set to either $0.25$ or $0.5$. First, it is noticed that the bounds and iterative estimates agree with the actual capacity for two-hop networks. Second, it is noticed that the bounds and estimates capture the variation of the actual capacity of the network. However, the estimates are more accurate. For both choices of channel parameters, both estimates predict throughput capacity within an error of 1\%. Further, it is also noticed from the figures that the independence assumptions of the estimates generally over-estimate the actual capacity of the network.

In order to study the effect of buffer size on capacity, we simulated a five-hop line network with each link having erasure probabilities just as in the previous setting. Figure~\ref{FB-LineM} presents the variation of our results and the 
\begin{figure}[ht!]
\centering
\psfrag{Xaxis}{\footnotesize{Buffer size $m$}}
\psfrag{Yaxis}{\hspace{-0mm}\footnotesize{Capacity (packets/epoch)}}
\psfrag{dddddddddddddddddata1}{\scriptsize{DbIE ($\ve=0.25$)}}
\psfrag{data2}{\scriptsize{RbIE ($\ve=0.25$)}}
\psfrag{data3}{\scriptsize{Capacity ($\ve=0.25$)}}
\psfrag{data4}{\scriptsize{LB ($\ve=0.25$)}}
\psfrag{data5}{\scriptsize{UB ($\ve=0.25$)}}
\psfrag{data6}{\scriptsize{DbIE ($\ve=0.50$)}}
\psfrag{data7}{\scriptsize{RbIE ($\ve=0.50$)}}
\psfrag{data8}{\scriptsize{Capacity ($\ve=0.50$)}}
\psfrag{data9}{\scriptsize{LB ($\ve=0.50$)}}
\psfrag{data10}{\scriptsize{UB ($\ve=0.50$)}}
\includegraphics[width=3.4in,angle=0]{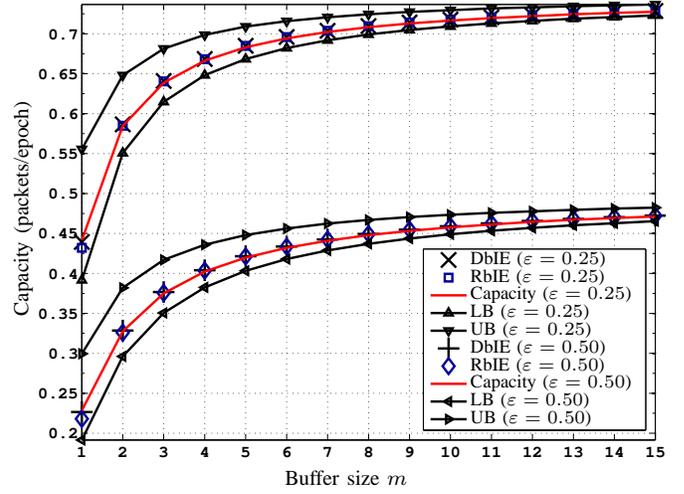}
\caption{Capacity of line networks with $h=5$ with varying buffer size $m$.}\label{FB-LineM}
\end{figure}
actual capacity as the buffer size of the intermediate node is varied. It can be seen that as the buffer size is increased, all curves approach the ideal min-cut capacity of $1-\ve$. Also, as is expected, the accuracy of the bounds improve with the buffer size.

Finally, the effect of the channel conditions on the capacity of a five-hop line network with intermediate buffer sizes of five packets each is presented in Figure~\ref{FB-LineE}. It is noticed that as the
\begin{figure}[ht!]
\centering
\psfrag{Xaxis}{\footnotesize{Probability of erasure $\ve$}}
\psfrag{Yaxis}{\footnotesize{Capacity (packets/epoch)}}
\psfrag{dddddddddddddddata1}{\footnotesize{UB}}
\psfrag{Data2}{\footnotesize{LB}}
\psfrag{Data3}{\footnotesize{Capacity}}
\psfrag{Data4}{\footnotesize{RbIE}}
\psfrag{Data5}{\footnotesize{DbIE}}
\psfrag{Data6}{\footnotesize{Min-cut Capacity}}
\includegraphics[width=3.4in,angle=0]{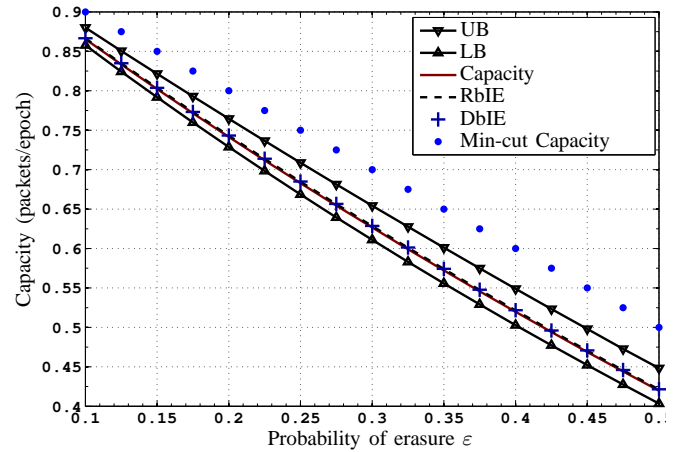}
\caption{Capacity of line networks with $h=5$, $m=5$ with varying erasure probability $\ve$.}\label{FB-LineE}
\end{figure}
probability of erasure increases, the loss in capacity due to finite buffer becomes more pronounced. For example, for the simulation setting of Fig.~\ref{FB-LineE}, the loss in capacity varies from $3.85\%$ at $\ve=0.1$ to $16.1\%$ at $\ve=0.5$ in a near-linear fashion. From these figures, we infer that it is paramount that the effect of blocking be considered as realistically as possible. Modeling the effect of blocking as packet loss (as is done to derive our bounds) only allows us to loosely bound the capacity of such networks.

\begin{figure}[ht!]
\centering
\psfrag{Xaxis}{\footnotesize{Delay (epochs)}}
\psfrag{Yaxis1}{\footnotesize{Probability mass}}
\psfrag{Yaxis}{\footnotesize{Difference in Estimates}}
\psfrag{ddddddddddddddddddddata}{\footnotesize{Sim. $m=5$}}
\psfrag{data2}{\footnotesize{Sim. $m=10$}}
\psfrag{data3}{\footnotesize{Sim. $m=15$}}
\psfrag{data4}{\footnotesize{RbIE $m=5$}}
\psfrag{data5}{\footnotesize{RbIE $m=10$}}
\psfrag{data6}{\footnotesize{RbIE $m=15$}}
\psfrag{data7}{\footnotesize{DbIE $m=5$}}
\psfrag{data8}{\footnotesize{DbIE $m=10$}}
\psfrag{data9}{\footnotesize{DbIE $m=15$}}
\psfrag{dddddddddddata1}{\footnotesize{$m=15$}}
\psfrag{data10}{\footnotesize{$m=10$}}
\psfrag{data11}{\footnotesize{$m=5$}}
\psfrag{(a)}{\small{(a)}}
\psfrag{(b)}{\small{(b)}}
\psfrag{m=5}{\tiny{$m=5$}}
\psfrag{m=10}{\tiny{$m=10$}}
\psfrag{m=15}{\vspace{2mm}\tiny{$m=15$}}
\includegraphics[height=2.5in,width=3.4in,angle=0]{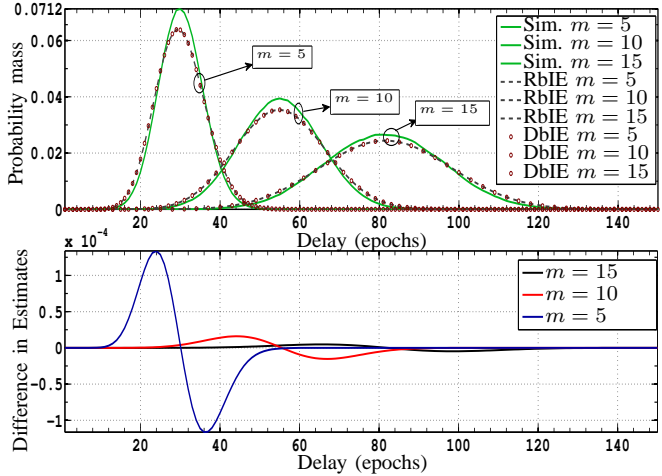}
\caption{FCFS delay profiles in 8-hop line networks with different buffers.}\label{Delaydist-Line}
\end{figure}

Figure~\ref{Delaydist-Line} presents the variation of delay profile for the optimal strategy in an eight-hop line network with the erasure probability on every link set to $0.25$. The delay profiles were simulated for three different buffer sizes. As in Section~\ref{FB-Delay}, the estimate and simulations were performed for the first-come first-serve strategy. From the first sub-plot, it is noticed that both mean and variance of the delay distribution increase as buffer sizes increase. While the mean delay obtained via simulations for the three memory settings are $30.22$, $55.18$, and $81.29$ epochs, whereas the analytical result for the same using the Dist.-based Iterative Estimate are $30.09$, $55.22$, and $81.68$ epochs, respectively. Note that the analytical estimates for the mean delay $\mu^*(\mce,\mcm)$ can be obtained without computing the delay profile by the use of Little's theorem~\cite{LittleRef} as follows.
\begin{equation}
\mu^*(\mce,\mcm)=\sum\limits_{i=1}^{h-1}\frac{ \langle \phi_{v_i}^* \rangle}{\mathcal{C}(\mce,\mcm)}=\sum\limits_{i=1}^{h-1}\frac{ \langle \phi_{v_i}^* \rangle}{\ol\ve_{h}(1-{\phi_{v_{h-1}}^*(0)})},
\end{equation}
where, as before, $\phi_{v_i}^*(\cdot)$ denotes the distribution of occupancy of $v_i$ at steady state given by the Rate-based Estimate. Note that each term in the above sum can be viewed as the contribution of the corresponding node to overall delay. It is noted that the analytic prediction of the delay profile is more conservative than the actual delay profile in the sense that the estimate of the variance is higher than the actual variance of packet delay. The second sub-plot of the figure illustrates the difference in the cumulative distribution of delay predicted by the two estimates. It is noticed from all the above simulations that there is only a minor difference between the two estimation schemes if the parameters of interest are either the throughput capacity or the delay profile.

Figure~\ref{FB-PDF} highlights the difference between the two estimates when continuous-time models are emulated using discrete-time epochs. Consider a three-hop line network where intermediate nodes have a buffer of three packets and their packet service distributions are exponential with $(\lambda_2,\lambda_3)=(3,2.99)\, s^{-1}$. Suppose that the arrival process at the first node is renewal with inter-arrival distribution being exponential with $\lambda_1=10\, s^{-1}$. The following figure presents the distribution of inter-departure duration from the second node. It is observed
\begin{figure}[ht!]
\centering
\psfrag{Xaxis}{\footnotesize{Time (s)}}
\psfrag{Yaxis}{\footnotesize{Probability density}}
\psfrag{ddddddddddddddddddata1}{\footnotesize{DbIE ($\Delta=0.001$)} }
\psfrag{data3}{\footnotesize{Continuous-time Simulation}}
\psfrag{data2}{\footnotesize{RbIE ($\Delta=0.001$)}}
\psfrag{1.5}{}
\includegraphics[height=2.35in,width=3.4in,angle=0]{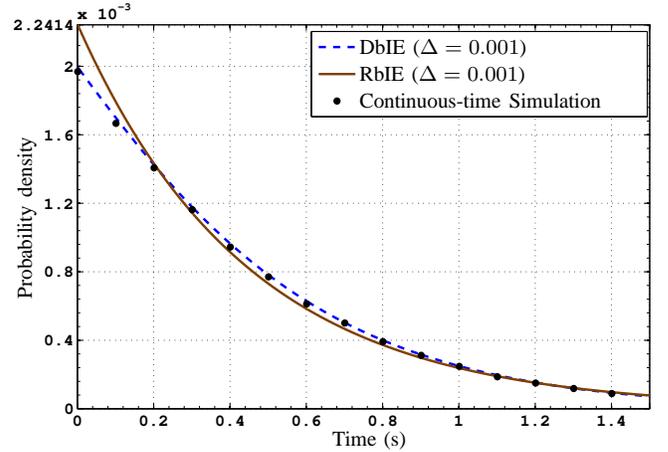}
\caption{The probability density of packet inter-arrival duration at the destination in a three-hop continuous-time line network.}\label{FB-PDF}
\end{figure}
that by lumping $\Delta=0.001$ seconds into each epoch, the Distribution-based Iterative Estimate provides a near-accurate distribution of the inter-departure durations. On the other hand, the Rate-based Iterative Estimate approximates the distribution as an exponential, which yields a less accurate estimate. Note that for this setting $\mathcal{C}_{\Lambda}(\mcm)=2.2467$ packets/sec, and the Distribution-based and Rate-based Estimates are 2.2447 and 2.2413 packets/sec, respectively.

\subsection{Buffer Allocation in Line Networks}\label{FB-Discussion}

In this section, we present a brief discussion on two questions pertaining to efficient usage of buffers in intermediate nodes. \emph{Is the use of more buffer slots, the merrier?} and \emph{How to allocate buffers to different nodes so that operation ensures near-min-cut throughput and acceptable delay?}

To address the first question, consider the eight-hop network of Fig.~\ref{Delaydist-Line}. As the buffer size is varied from 10 to 15 packets, the Rate-based Estimate for capacity changes from 0.7135 to 0.7254 packets/epoch -- a change of less than 1.5\% (of the min-cut bound). However, the mean latency changes from 55.18 to 81.29 epochs -- a 47\% change. Therefore, for each $\mce$, it is likely that there is a critical buffer size for each node beyond which the throughput capacity improvement is marginal; however, with increase in buffer sizes, the average time packets spend in the network continues to grows significantly. One must therefore identify the correct size of buffers to be used so that both latency and throughput capacity are acceptable.

To discuss the second issue, we illustrate with the following example. Consider a four-hop network with $\mce=[0.3\,\, 0.5\,\, 0.5\,\, 0.2]$ for which a good choice of buffer allocation needs to be identified under the constraint that the total number of buffers in the network must be no more than 30 packets. To this end, we use the Rate-based Estimate to study
\begin{figure}[htbp!]
\centering
\psfrag{data1}{\small{$v_1$}}
\psfrag{data2}{\small{$v_2$}}
\psfrag{data3}{\small{$v_3$}}
\psfrag{Xaxis}{\small{Buffer Size}}
\psfrag{Yaxis1}{\small{Capacity Estimate}}
\psfrag{Yaxis}{\small{Delay at a Node}}
\includegraphics[height=2.5in,width=3.4in, angle=0]{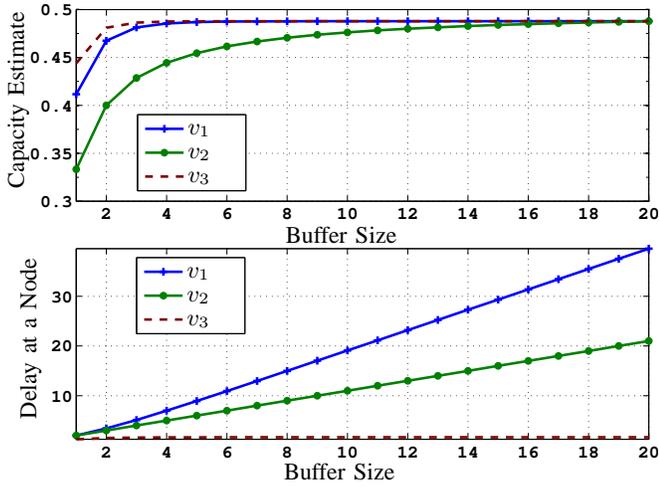}
\caption{Throughput, and average delay contribution by each intermediate node with varying buffer size.}
\label{Delay-Throughput}
\end{figure}
the effect of individual buffer sizes on throughput and delay. Fig.~\ref{Delay-Throughput} shows the variation of the throughput and delay contributed by each node when its memory is varied from 1 to 20 packets, while the buffer sizes of other intermediate nodes are kept at 20 packets. In this example, it is noticed that maximum throughput estimate for all choices of memory estimates is $0.4871$ packets/epoch when $\mcm_a=(5, 21, 4)$. This setting offers a mean packet delay of $32.24$ epochs. However, minimum delay configuration amongst those that offer a throughput more than $0.485$ packets/epoch is $\mcm_b=(4, 20, 6)$, which offers a throughput of $0.4851$ packets/epoch and a mean packet latency of 28.46 epochs. The actual capacity and delay for these configurations were found to be $\mcc(\mce,\mcm_a)=0.4871$ packets/epoch, $\mu(\mce,\mcm_a)=32.17$ epochs and $\mcc(\mce,\mcm_b)=0.4858$ packets/epoch, $\mu(\mce,\mcm_b)= 28.33$ epochs, respectively.

To understand further these patterns, we present in Fig.~\ref{Buffer-Occupancy} the steady-state occupancy of the three intermediate nodes when buffer sizes are set to $\mcm_a=(5,21,4)$ packets, $\mcm_b=(4,20,6)$ packets and $\mcm_c=(15,15,15)$ packets, respectively. In all settings, it is noted that the node $v_1$ is congested because the sub-network from $v_1$ to $v_4$ has a min-cut
\begin{figure}[htbp!]
\centering
\psfrag{dat1}{\footnotesize{$\mcm_b$}}
\psfrag{dat2}{\footnotesize{$\mcm_a$}}
\psfrag{dat3}{\footnotesize{$\mcm_c$}}
\psfrag{Yaxis}{\footnotesize{Probability}}
\psfrag{Xaxis1}{\footnotesize{Occupancy for $v_1$}}
\psfrag{Xaxis2}{\footnotesize{Occupancy for $v_2$}}
\psfrag{Xaxis3}{\footnotesize{Occupancy for $v_3$}}
\includegraphics[width=3.4in,angle=0]{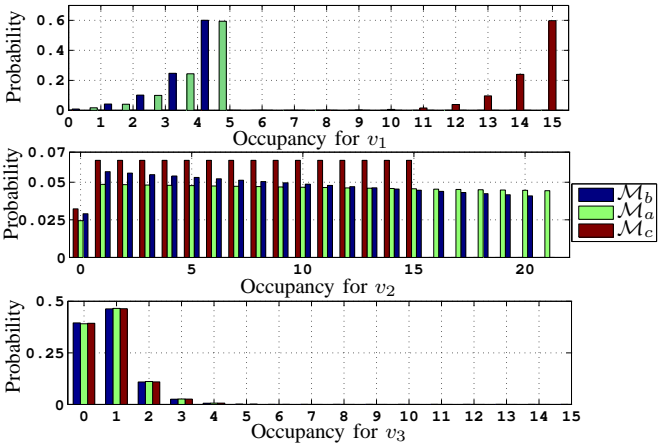}
\caption{Estimated Buffer occupancy distribution in intermediate nodes.}
\label{Buffer-Occupancy}
\end{figure}
capacity of $0.5$, whereas it receives packets at the rate of $0.7$. Therefore, the steady-state occupancy of the node $v_1$ for $\mcm_a$ and $\mcm_c$ are translates of that of $\mcm_b$. Due to congestion, an arriving packet at such a node usually sees very high occupancy. Hence, in a first-come first-serve mode of operation, the arriving packet has to wait long before getting serviced. Therefore, it is critical that the buffer size of congested nodes (such as $v_1$) be kept to absolute minimum to minimize average packet delay. Similarly, $v_3$ can at most receive packets at a rate of $0.5$, however the outgoing link can communicate packets at a much higher rate. Therefore, the buffer of $v_3$ is never full as long as the buffer size is greater than five. Nodes such as $v_3$ that are never congested contribute little to the delay experienced by packets. Hence, limiting buffer sizes of such nodes is not critical for delay as long as the sizes are bigger than their threshold sizes (beyond which throughput increase is marginal).

Occupancy in nodes like $v_2$ that are neither congested nor starved undergo non-trivial changes with changes in buffer sizes. These nodes contribute significantly to both the throughput and average packet delay in the network. For example, in the example network $v_2$ has a near-uniform distribution for both $\mcm_a$ and $\mcm_b$. Just like congested nodes, such nodes have to be allocated buffer sizes so that the they neither block packets nor contribute to delay significantly. Though the classification of nodes as congested, starved or neither can usually be done by focusing on $\mce$, good memory allocation requires knowledge of trends of latency and throughput with buffer sizes, which in turn require the help of more sophisticated estimates such as those proposed in this work.

As a second example, consider another four-hop network with $\mce_c=[0.51\,\,0.50\,\,0.49\,\,0.48]$. In the infinite buffer setting, the queueing system corresponding to this buffer configuration is stable. Hence, no node can be classified \emph{a priori} as congested. Suppose that a throughput-optimal allocation of buffer sizes for intermediate nodes is to be designed with the constraint that the total number of packets in the network be limited to 60. Clearly, a na\"{\i}ve first guess is to assign $\mcm_d=[20\,\,20\,\,20]$. However, notice that no matter how large the buffer sizes are, the probability of blocking at any node is always non-zero. Hence, the rate of arrival that $v_2$ and $v_3$ see is smaller than that noticed by $v_1$. Therefore, it is meaningful to assign $v_1$ a larger buffer size to minimize blocking at $v_1$ and maximize throughput. Although this intuition is correct, it is unclear as to how to allocate buffers. The strength of the iterative technique is in resolving exactly this issue by assigning estimates to each buffer allocation configuration. By searching around the neighborhood of $\mcm_d$, the maximum throughput configuration is found to be $\mcm_e=[27\,\,20\,\,13]$.

As is illustrated by these examples, the proposed iterative estimation techniques presents a framework to identify nodes in line networks that are either: (a) starved and therefore play an insignificant role in capacity and packet delay (such as $v_3$ of Fig.~\ref{Buffer-Occupancy}), or (b) congested and contribute significantly to packet delay (such as $v_1$ of Fig.~\ref{Buffer-Occupancy}), or (c) contribute significantly to both capacity and packet delay (such as $v_2$ of Fig.~\ref{Buffer-Occupancy}). On identifying these nodes, it is possible to identify configurations that make efficient use of the buffers without severely compromising on either throughput capacity or average packet delay.

\section{Conclusions}\label{Conc}

This work focused on the effect of finite buffers on the throughput capacity and packet delay profile in line networks with packet erasure links. First, an exact Markovian framework for modeling line networks with perfect feedback was presented. The framework was simplified using independence assumptions to derive iterative estimation techniques that yield approximations of all marginal buffer statistics and also allow to identify the packet delay profile in such networks. Further, it was shown that the absence of feedback has no effect on the throughput capacity of line networks provided packet-level coding is permitted. Finally, via simulations, the proposed iterative techniques were noticed to be computationally-efficient and near-accurate models to analyze and study the behavior of line networks.

\appendices
\section{Discrete and Continuous models}\label{App.0-DiscvsCts}
In this section, we argue that the discrete model assumed in the paper can be used to study the capacity of tandem queue model with type II blocking (see~\cite{HJPBook}) and independent exponential service times at each node. Consider a tandem queue of $h$ links and $h-1$ intermediate nodes. Suppose $\Lambda=(\lambda_1,\ldots,\lambda_h)$ denotes the parameters for the exponential service times at $v_0,\ldots,v_{h-1}$, respectively. Assuming that each intermediate node has buffers given by $\mcm=(m_1,\ldots,m_{h-1})$ and that $N_t$ denotes the number of packets collected by $v_h$ in the period $[0,t)$, the throughput capacity $\mcc_\Lambda(\mcm)$\footnote{Note that the definition of throughput using (\ref{eqn40}) hinges on the ergodicity of the continuous-time system.} of the system is defined by
\begin{equation}
\mcc_\Lambda(\mcm) \triangleq \lim_{t\rightarrow \infty} \frac{N_t}{t} \label{eqn40}
\end{equation}
can be computed from a discrete model assumed in this paper. For example,
\begin{figure}[htbp!]
\centering
\psfrag{ddddddddata1}{\small{$\mcc_{\Lambda}$}}
\psfrag{data2}{\small{$\tau=1/4$}}
\psfrag{data3}{\small{$\tau=1/8$}}
\psfrag{data4}{\small{$\tau=1/16$}}
\psfrag{data5}{\small{$\tau=1/64$}}
\psfrag{Memory}{\small{Buffer sizes}}
\psfrag{Throughput}{\small{Throughput (in packets/sec)}}
\includegraphics[width=3.4in, angle=0]{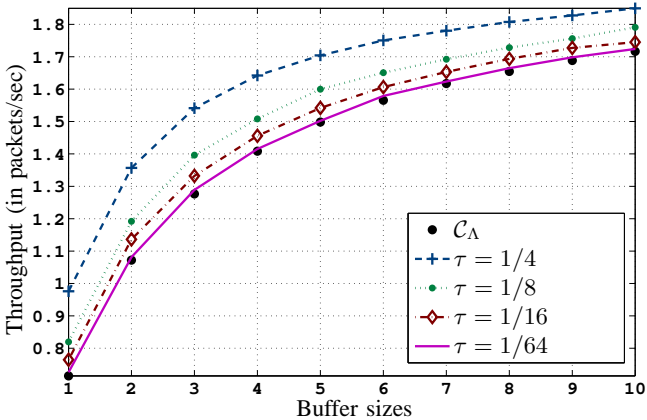}
\caption{Throughput of continuous and discrete systems for varying buffer sizes}
\label{CtsvsDisc-Throughput}
\end{figure}
Fig.~\ref{CtsvsDisc-Throughput} considers a four-hop system with each node having exponential processing times with parameter $\lambda=2$ s${}^{-1}$ and compares it with four discretized models. Note that the approximations become finer as smaller values of $\tau$ are chosen. This fact can be formalized as follows.

\begin{thm}
\begin{equation}
\mcc_\Lambda(\mcm) = \lim_{\tau\rightarrow 0} \tau^{-1}\mcc^{PF}(1-\Lambda\tau,\mcm),
\end{equation}
where the right-hand side uses the discrete-time model of (\ref{FB-eqn5.1}).
\end{thm}
\begin{proof}
We begin by constructing the probability transition matrices for continuous and discrete chains that track the state of the system just before a departure from the last intermediate node. Note that both chains use the same state space $\mcs=\{(s_1,\ldots,s_{h-1}): 0\leq s_i \leq m_i\}$, however their transition probabilities are different.

Let $\Pi$ denote the probability transition matrix for the continuous model. Let $M$ denote the transition matrix that effects the change in states when a departure from the last intermediate node occurs and let $P_t$ denote the transition matrix corresponding to changes in state over a duration of $t$ seconds given that no departure occurs in that duration. Then,
\begin{equation}
\Pi=M\int_\mathbb{R} P_t dF_{\lambda_{h}}(t),
\end{equation}
where $F_{\lambda_{h}}$ denotes the cumulative density function of the exponential RV with parameter $\lambda_{h}$. Notice that for any $\tau>0$, $P_t=(P_\tau)^{\lfloor \frac{t}{\tau} \rfloor} P_{t-\tau\lfloor \frac{t}{\tau}\rfloor}$. Therefore,
\begin{align}
P_t=\lim_{\tau\rightarrow 0} (P_\tau)^{\lfloor \frac{t}{\tau} \rfloor}  P_{t-\tau\lfloor \frac{t}{\tau}\rfloor} \stackrel{(a)}{=} \lim_{\tau\rightarrow 0} (P_\tau)^{\lfloor \frac{t}{\tau} \rfloor},
\end{align}
where (a) follows since $P_{t-\tau\lfloor \frac{t}{\tau}\rfloor} \rightarrow I$. Let $P^\Delta_\tau$ denote the state transition matrix for one time epoch of the discretized model with $\mce=1-\Lambda \tau$. Then, we have
\begin{align}
P_\tau=P^\Delta_\tau +o(\tau^2) \Rightarrow P_t &= \lim_{\tau\rightarrow 0} \Big(P^\Delta_\tau+o(\tau^2)\Big)^{\lfloor \frac{t}{\tau}\rfloor}\nonumber\\
&=\lim_{\tau\rightarrow 0} \Big(P^\Delta_\tau\Big)^{\lfloor \frac{t}{\tau}\rfloor}.
\end{align}
Therefore, we have
\begin{align}
\Pi&=M\int_\mathbb{R} P_t dF_{\lambda_{h}}(t) = M\int_\mathbb{R} \lim_{\tau\rightarrow 0} \Big(P^\Delta_\tau\Big)^{\lfloor \frac{t}{\tau}\rfloor} dF_{\lambda_{h}}(t)\nonumber\\
&\stackrel{(a)}{=} \lim_{\tau\rightarrow 0} M\int_\mathbb{R} \Big(P_\tau^\Delta\Big)^{\lfloor \frac{t}{\tau}\rfloor} dF_{\lambda_{h}}(t)\nonumber\\
&= \lim_{\tau\rightarrow 0} \sum_{i=0}^\infty M{P^\Delta_\tau}^i \int_{i\tau}^{i\tau+\tau} dF_{\lambda_{h}}(t)\nonumber\\
&= \lim_{\tau\rightarrow 0} \Bigg(\sum_{i=0}^\infty \bigg[\Big(e^{-i\lambda_{h}\tau}-e^{-(i+1)\lambda_{h}\tau}\Big) M{P^\Delta_\tau}^i \bigg]\Bigg),
\end{align}
where (a) follows from Fubini-Tonelli Theorem~\cite{FollandBook}. The discrete equivalent $\Pi^\Delta_\tau$ of the above transition matrix that tracks the state between departures for the corresponding time-discretized system is given by
\begin{align}
\Pi_\tau^\Delta&=\lambda_{h}\tau M+(\ol{\lambda_{h}\tau})\lambda_{h}\tau MP^\Delta_\tau+(\ol{\lambda_{h}\tau})^2\lambda_{h}\tau M{P^\Delta_\tau}^2+\cdots\nonumber\\
&= \sum_{i=0}^\infty \bigg[\Big((\lambda_{h}\tau)(\ol{\lambda_{h}\tau})^{i}\Big)M{P^\Delta_\tau}^i\bigg]
\end{align}
Let $\xi(\mathbf{s},\mathbf{s}')\triangleq\lim_{\tau\rightarrow 0} \Big(\Pi(\mathbf{s},\mathbf{s}')-\Pi^\Delta_\tau(\mathbf{s},\mathbf{s}')\Big)$ for any pair of states $\mathbf{s}$, $\mathbf{s}'$. Then,
\begin{align}
\xi(\mathbf{s},\mathbf{s}')&\stackrel{(b)}{\leq} \lim_{\tau\rightarrow 0} \bigg[\sum_{i=0}^\infty \Big|e^{-i\lambda_{h}\tau}(\ol{e^{-\lambda_{h}\tau}}) -(\lambda_{h}\tau)(\ol{\lambda_{h}\tau})^{i}\Big|\bigg]\nonumber\\
&\stackrel{(c)}{\leq}\lim_{\tau\rightarrow 0} \bigg[\sum_{i=0}^\infty (\lambda_{h}\tau)\Big|e^{-i\lambda_{h}\tau}-(\ol{\lambda_{h}\tau})^{i}\Big| +\frac{e(\lambda_{h}\tau)^2}{\ol{e^{-\lambda_{h}\tau}}}\bigg]\nonumber\\
&=\lim_{\tau\rightarrow 0} \bigg[\sum_{i=0}^\infty (\lambda_{h}\tau)\Big|e^{-i\lambda_{h}\tau}-(1-\lambda_{h}\tau)^{i}\Big|\bigg]\nonumber\\
&\stackrel{(d)}{=}\lim_{\tau\rightarrow 0} \bigg[(\lambda_{h}\tau)\sum_{i=0}^\infty \Big(e^{-i\lambda_{h}\tau}-(1-\lambda_{h}\tau)^{i}\Big)\bigg]\nonumber\\
&=\lim_{\tau\rightarrow 0} \Big(\frac{\lambda_{h}\tau}{1-e^{-\lambda_{h}\tau}}-1\Big)=0.
\end{align}
Note that (b) follows since $M{P^\Delta_\tau}^i$ is a probability matrix and hence each component is bounded above by unity, and (c) follows since for $\lambda_{h}\tau<1$, it is true that
\begin{equation}
|1-e^{\lambda_{h}\tau}|-{\lambda_{h}\tau}=\sum_{i\geq 2}\frac{(\lambda_{h}\tau)^i}{i!}\leq e(\lambda_{h}\tau)^2,\nonumber
\end{equation}
and (d) follows from $e^{-x} >1-x$ for $x>0$. Thus $\Pi^\Delta_\tau \rightarrow \Pi$ as $\tau\rightarrow 0$. Let $\nu^\Delta_\tau$ and $\nu$ be the eigenvectors of $\Pi_\tau^\Delta$ and $\Pi$, respectively. Then, since the steady-state distribution of a chain is a continuous function of the transition matrix, it follows that $\nu^\Delta_\tau \rightarrow \nu$, as $\tau\rightarrow 0$. However, the capacity computed using continuous and discrete models are given by
\begin{align}
\mcc_\Lambda(\mcm)&={\lambda_{h}}\sum_{\mathbf{s}:s_{h-1}>0}\nu(\mathbf{s})\nonumber\\
\mcc^{PF}(1-\Lambda\tau,\mcm)&={\lambda_{h}\tau}\sum_{\mathbf{s}:s_{h-1}>0}\nu_\tau^\Delta(\mathbf{s})\nonumber.
\end{align}
Therefore, $\tau^{-1}\mcc^{PF}(1-\Lambda\tau,\mcm)\rightarrow \mcc_\Lambda(\mcm)$.
\end{proof}

\section{Proof of Lemma~\ref{FB-lem1}}\label{App.1-1}
 (a) Suppose that the state of the system is $\mathbf{n}=(n_1(l),\ldots,n_{h-1}(l))$ with $0< n_{h-1}(l) < m_{h-1}$, then from (\ref{FB-eqn2}), we notice that $Y_h=X_h$ and $Y_{h-1}=\sigma[n_{h-2}(l)]X_{h-1}$. Hence, given the event $m_{h-1}>n_{h-1}(l)>0$, $\mathbf{Y}(l)=(Y_1(l),\ldots,Y_{h}(l))$ depends only on $(n_1(l),\ldots,n_{h-2}(l))$ and $(X_1(l),\ldots,X_{h}(l))$ and not on $n_{h-1}(l)$. This guarantees that $(\Gamma_i^-, \Gamma_i^+, \Omega_i)= (\Gamma_j^-, \Gamma_j^+, \Omega_j)$ for $0< i,j < m_{h-1}$.

 (b) First suppose $h>2$. Consider $\Gamma_i^-$ for some $i>0$ and the state of the system at some time $l\in\mathbbm{N}$. $\Gamma_i^-$ represents transitions from states that have the form   $(n_1(l),n_2(l),\ldots,n_{h-1}(l)=i)$ to states of the form   $(n_1(l+1),n_2(l+1),\ldots,n_{h-1}(l+1)=i-1)$. Since $n_{h-1}(l+1)=n_{h-1}(l)-1$, it must be that $Y_{h-1}(l)=0$ and that the channel must have erased the packet transmitted by $v_{h-2}$. Denote $L_i=\prod_{1\leq k<i}(m_i+1)$ for $i>1$ and $L_1=1$. Then, it is seen that for any realization of   $\{X_i(l)\}_{i=0}^{h-3}$, it is true that the state transition must obey
 \begin{equation}
 1+n_1(l)+\sum_{i=2}^{h-2} n_i(l)L_i\leq 1+n_1(l+1)+\sum_{i=2}^{h-2} n_i(l+1)L_i.\nonumber
 \end{equation}
 However, $\Bigl(1+n_1(l)+\sum_{i=2}^{h-2} n_i(l)\prod_{j=1}^{i-1} (m_j+1)\Bigr)$ is the index of the row corresponding to the state $\bfn(l)$ within $\Gamma_i^-$ and $\Bigl(1+n_i(l+1)+\sum_{i=2}^{h-2} n_1(l+1)\prod_{j=1}^{i-1} (m_j+1)\Bigr)$ is the index of the column corresponding to $\bfn(l+1)$ within $\Gamma_i^-$. Therefore, all possible transitions in $\Gamma_i^-$ correspond to transitions from states to other state that involve a non-positive change in the row-index. Therefore, $\Gamma_i^-$ is upper triangular. Finally, since each diagonal
 term of $\Gamma_i^-$ is bounded below by $\overline{\ve}_{h}\prod_{k=0}^{h-2} \varepsilon_{k+1}$, we conclude that
 \begin{equation}
\det(\Gamma_i^-)\geq \Bigl(\overline{\ve}_{h}\prod_{k=0}^{h-2} \varepsilon_{k+1}\Bigr)^{L_{h-1}}>0.
 \end{equation}
Finally, if $h=2$, it is easy to see that $\Gamma^-_i=[\ol\ve_2\ve_1]$.

(c) Consider a transition under $\Gamma_i^+$ for $i<m_{h-1}$ from a state that has the form
$(n_1(l),n_2(l),\ldots,n_{h-1}(l)=i)$ to another that has the form $(n_1(l+1),n_2(l+1),\ldots,n_{h-1}(l+1)=i+1)$ after an epoch. Since $n_{h-1}(l+1)=n_{h-1}(l)+1$, it must be that the packet transmitted during this epoch on the link $(v_{h-2},v_{h-1})$ must have reached successfully, {i.e.}, $Y_{h-2}(l)=1$. By an argument similar to the above one, we
can show that $\Gamma_i^+$ is lower triangular. However, certain diagonal terms are zero. In specific, consider the transition from state $(n_1(l)=0,\ldots,n_{h-2}(l)=0, n_{h-1}(l)=i)$ to the state $(n_1(l)=0,\ldots,n_{h-2}(l)=0, n_{h-1}(l)=i+1)$ which corresponds to the $(\Gamma_i^+)_{11}$. However, this transition is impossible when $h>2$, since the node $v_{h-2}$ has no packets to send during this epoch. Thus, $\det(\Gamma_i^+)=0$ if $h\geq 3$.

(d) The non-singularity of $I-\Omega_i$ follows from the fact that $(I-\Omega_i)$ is diagonal dominant~\cite{Horn}, since $(I-\Omega_i)_{kk}\geq\sum_{k'\neq k} |(I-\Omega_i)_{kk'}|$. On the other hand, since $\Gamma_i^+,\Gamma_i^-\neq \mathbf{0}$, there exists at least one $k$ for which the inequality is strict, which guarantees the non-singularity of these matrices.

\section{Proof of Theorem~\ref{FB-thm1}}\label{App.1-2}
We proceed by mathematical induction on the time index $l$. Clearly, the condition holds for $l=0$. Suppose that the claim is true for all nodes and for times $l=0,\ldots,k$ for some $k\geq 0$. Consider the states of the node $v_i$ for some $i=2,\ldots,h-1$ in both chains at time instant $k$. One of the two following cases must apply.
\begin{itemize}
\item[1.] $\underline{n_i(k)=\nt_i(k)}$: In this case, we note that
\begin{equation}
n_i(k+1)-\nt_i(k+1)=Y_i(k)-\tilde{Y}_i(k)-Y_{i+1}(k)+\tilde{Y}_{i+1}(k). \nonumber
\end{equation}

If $n_i(k)=\nt_i(k)=0$, then $Y_{i+1}(k)=\tilde{Y}_{i+1}(k)=0$ and
\begin{align}
Y_i(k)-\tilde{Y}_i(k)=X_i(k)[\sigma[n_{i-1}(k)]-\sigma[\nt_{i-1}(k)]\geq 0.\nonumber
\end{align}
Thus $n_i(k+1)-\nt_i(k+1)\geq 0$.

Now, if $n_i(k)=\nt_i(k)=m_i$, it is seen from (\ref{FB-eqn2}) and (\ref{FB-eqn6}) that $\tilde{Y}_{i+1}(k)-Y_{i+1}(k)\geq 0$. Further, if $Y_{i+1}(l)=0$, then clearly, $n_i(k+1)=m_i$ and $\nt_i(k+1)\leq m_i = n_i(k+1)$.

If $Y_{i+1}(k)=1$, then $\tilde{Y}_{i+1}(k)=1$ and $\nt_i(k+1)\leq n_i(k+1)$ follows since
\begin{align}
Y_i(k)-\tilde{Y}_i(k)=X_i(k)\big(\sigma[n_{i-1}(k)]-\sigma[\nt_{i-1}(k)]\big)\geq 0.\nonumber
\end{align}

Now, if $0<n_i(k)=\nt_i(k)<m_i$, then (\ref{FB-eqn2}) and (\ref{FB-eqn6}) again imply $\tilde{Y}_{i+1}(k)-Y_{i+1}(k)\geq 0$ and $Y_i(k)-\tilde{Y}_i(k)=X_i(k)\big(\sigma[n_{i-1}(k)]-\sigma[\nt_{i-1}(k)]\big)\geq 0$, and hence $n_i(k)\geq \nt_i(k)$ follows.

\item[2.] $\underline{n_i(k)\geq\nt_i(k)+1}$: Assume let $\nt_i(k)>0$. Then,
\begin{align}
\nt_i(k+1)&\stackrel{(\ref{AMCeqn})}{\leq}\nt_i(k)+1-\tilde{Y}_{i+1}(k)\nonumber\\
&=\nt_i(k)+ 1 - \sigma[\nt_i(k)]X_i(k)\nonumber\\
&\leq\nt_i(k)+ 1 - Y_{i+1}(k) \leq n_i(k)-Y_{i+1}(k)\nonumber\\
&\leq n_i(k)+Y_i(k)-Y_{i+1}(k) = n_i(k+1).\nonumber
\end{align}
Lastly, if $\nt_i(k)=0$, then the claim can be violated only if $\nt_i(k+1)=1$ and $n_i(k+1)=0$, which can happen only if $X_i(k)=X_{i+1}(k)=1$. However, under this channel instance, $n_i(k+1)\geq n_i(k)\geq 1$. Thus, $n_i(k_1)\geq \nt(k+1)$.
\end{itemize}
Thus, we have the following.
\begin{equation}
n_i(k)\geq \nt_i(k), \quad \quad i=2,\ldots h-1.
\end{equation}
The proof is then complete by following the above argument for $v_1$ and interpreting $\sigma[n_{0}(k)]=\sigma[\nt_{0}(k)]=1$, since the source always possesses innovative packets.

\section{Proof of Theorem~\ref{FB-Ubthm}}\label{App.1-4}

If $h=2$ comparing (\ref{FB-eqn2}) and (\ref{FB-eqn6}), we see that the AMC and the EMC are identical. Hence, we may assume $h>2$. The proof in this case is based on mathematical induction on the time index $l$. At each time, we compare the state of the EMC with that of the modified AMC. Let the extended state of the EMC at an instant $l\in\mathbb{Z}_{\geq 0}$ be denoted by $\bfn^e(l)=(n_1(l),\ldots,n_h(l))$, where the notation is identical to that of Sec.~\ref{FB-sec2} with the addition that $n_h(l)$ denotes the number of packets that the destination has received by the $l^{\textrm{th}}$ epoch. Similarly define
the extended state of the AMC with modified buffer sizes at an instant $l\in\mathbb{Z}_{\geq 0}$ by $\bfq^e(l)$. Define a partial ordering of vectors of $\mathbb{Z}_{\geq 0}^h$ in the following manner. For two vectors $\mathbf{v},\mathbf{v'}\in \mathbb{Z}_{\geq 0}^h$, $\mathbf{v}\succeq\mathbf{v'}$ if $\sum_{k=i}^h v_k \geq \sum_{k=i}^h v'_k$ for each $i=1,\ldots,h$. We track the system starting from initial rest (all buffers being empty) using an instance of channel realizations. Clearly
$\bfq^e(0)\succeq\bfn^e(0)$.

Suppose that $\bfq^e(l)\succeq\bfn^e(l)$ for $l=0,\ldots k-1$. Consider $l=k$. One of the following two situations may arise\footnote{For convenience, we set $q_0(k)=n_0(k)\triangleq\infty$, $k\geq 0$ in this proof.}.
\begin{itemize}
\item [1.] $\ul{\{i<h:q_i^e(k-1)=\sum_{j=1}^i m_i\}=\emptyset}$: In this case, no node is saturated in the AMC and hence every node can potentially accept packets provided both the node preceding it has packets to send and the channel allows it. Consider the number of packets that are in the buffers of nodes $v_j,\ldots,v_h$ for some $0<j\leq h$ in both chains.

    (i) If $n_{j-1}(k-1)=0$ or if both $n_{j-1}(k-1)>0$ and $X_j(k-1)=0$ are true, then
\begin{equation}
\sum_{s=j}^hq_s^e(k)= \sum_{s=j}^hq_s^e(k-1)\geq \sum_{s=j}^h n_s^e(k-1)= \sum_{s=j}^h n_s^e(k).\nonumber
\end{equation}

 (ii) If $n_{j-1}(k-1)>0$, $X_j(k-1)=1$ and $q_{j-1}(k-1)=0$ then
\begin{align}
\sum_{s=j-1}^h q_s^e(k-1)&\geq \sum_{s=j-1}^h n_s^e(k-1)\nonumber\\ \Rightarrow  \sum_{s=j}^h q_s^e(k-1)&\geq \sum_{s=j}^h n_s^e(k-1) +n_{j-1}(k-1)\nonumber\\
&\geq \sum_{s=j}^h n_s^e(k-1)+1. \nonumber
\end{align}
Therefore,
\begin{equation}
\sum_{s=j}^h q_s^e(k)=\sum_{s=j}^h q_s^e(k-1) \geq \sum_{s=j}^h n_s^e(k-1)+1 \geq \sum_{s=j}^h n_s^e(k).\nonumber
\end{equation}

(iii) Finally, if $n_{j-1}(k-1)>0$, $X_j(k-1)=1$ and $q_{j-1}(k-1)>0$ then
\begin{align}
\sum_{s=j}^h q_s^e(k)&=\sum_{s=j}^h q_s^e(k-1)+1 \nonumber\\
&\geq \sum_{s=j}^h n_s^e(k-1)+1 \geq \sum_{s=j}^h n_s^e(k).\nonumber
\end{align}
 Since $j$ was arbitrary, it follows that $\bfq^e(k)\succeq \bfn^e(k)$.

\item [2.] $\ul{\{i<h:q_i^e(k-1)=\sum_{j=1}^i m_i\}\neq\emptyset}$: Then, let $I=\max\{i<h:q_i^e(k-1)=\sum_{j=1}^i m_i\}$. In this case, nodes $v_{I+1},\ldots,v_h$ are not saturated and can accept packets. The argument for $\sum_{s=j}^h q^e_s(k)\geq \sum_{s=j}^h n^e_s(k)$ follows for $j=I+1,\ldots,h$ is similar to the previous case. Notice that since the occupancy of nodes $v_{i}$, $i>I$ are not full,
 \begin{align}
\sum_{s\geq I} q_s^e(k)\geq \sum_{s> I} q_s^e(k-1)+\sum_{1\leq \iota\leq I} m_\iota. \label{App1-Comment}
\end{align}
Now, for $j=I$, two cases may occur.

(i) If $j=I>1$, then by (\ref{App1-Comment}),
\begin{align}
\sum_{s\geq I} q^e_s(k)&\geq \sum_{s> I} q_s^e(k-1)+\sum_{1\leq \iota\leq I} m_\iota\nonumber\\
&\geq \sum_{s>I} q_s^e(k-1)+m_I + 1
\nonumber\\
 &\geq \sum_{s\geq I} n_s^e(k-1) + 1 \geq \sum_{s\geq I} n_s^e(k).
\end{align}

(i) If $j=I=1$, and $X_1(k-1)=0$ then
\begin{align}
\sum_{s\geq 1} q^e_s(k)=\sum_{s\geq 1} q^e_s(k-1)\geq\sum_{s\geq 1} n^e_s(k-1)=\sum_{s\geq 1} \nonumber q^e_s(k).
\end{align}
However, if $j=I=1$ and $X_1(k-1)=1$, then
\begin{align}
\hspace{-1.5mm}\sum_{s\geq 1} q^e_s(k)&\geq \sum_{s> 1} q_s^e(k-1)+m_1+X_2(k)\nonumber\\
&\geq  \sum_{s> 1} n_s^e(k-1)+n_1^e(k)+X_2(k-1)\nonumber\\
&\geq \sum_{s> 1} n_s^e(k)\hspace{-0.45mm}-\hspace{-0.45mm}Y_2(k\hspace{-0.45mm}-\hspace{-0.45mm}1)+n_j^e(k)+X_2(k\hspace{-0.45mm}-\hspace{-0.45mm}1)\nonumber\\
&\geq \sum_{s\geq 1} n_s^e(k).
\end{align}
Thus, the claim holds for $j=I,I+1,\ldots,h$. The claim is then complete if $I=1$. Therefore, in what follows, we may assume $I>1$.

Finally, for $1<j<I$, one of the following cases must hold.

(i) If $X_1(k-1)=0$, then
\begin{align}
\sum_{s=j}^h q^e_s(k)&\geq \sum_{s=I}^h q^e_s(k) \stackrel{\ref{App1-Comment}}{\geq} \sum_{s=I+1}^h q^e_s(k-1)+ \sum_{\iota=1}^I m_\iota\nonumber\\
 &\geq \sum_{s=I+1}^h n^e_s(k-1)+\sum_{s=1}^I n^e_s(k-1)\nonumber\\
&= \sum_{s=1}^h n^e_s(k-1)=\sum_{s=1}^h n^e_s(k)\geq \sum_{s=j}^h n^e_s(k)\nonumber
\end{align}

(ii) If $X_{1}(k-1)=1$ then $q^e_1(k)\geq1$ and
\begin{align}
\sum_{s\geq j} q^e_s(k) &= \sum_{s\geq I} q^e_s(k) + \sigma[2-j]q^e_1(k)\nonumber\\
&\geq \sum_{s> I} q^e_s(k-1)+\sum_{\iota=1}^I m_\iota+\sigma[2-j]\nonumber\\
&\geq\sum_{s>I} n^e_s(k-1)+\sum_{s=j}^I n^e_s(k-1)+1\nonumber\\
&=\sum_{s\geq j} n^e_s(k-1)+1\geq \sum_{s\geq j} n^e_s(k).
\end{align}
Thus, the claim is true for all indices $j=1,\ldots,h$ and $\bfq^e(k)\succeq \bfn^e(k)$. Here, it must be noted that if the buffer sizes for the nodes of AMC are not modified as in the hypothesis, (\ref{App1-Comment}) will not hold.
\end{itemize}
Finally, the upper bound follows since
\begin{align}
\ol\mcc^{PF}(\mce,\mcm)&=\lim_{l\rightarrow\infty} \frac{q^e_h(l)}{l} \geq \lim_{l\rightarrow\infty} \frac{n^e_h(l)}{l}\nonumber= \mcc^{PF}(\mce,\mcm).\nonumber
\end{align}

\section{Proof of Theorem~\ref{FB-uniqthm}}\label{App.1-3}

Consider two rate-approximate solutions $(\mcr^a,\mcp^a)$ and $(\mcr^b,\mcp^b)$ such that $r_{h}^a=r_{h}^b=\delta$ with $0<\delta<1$. Notice that $\vp(0|r_{h-1},\ve_{h},0)$ is a strictly decreasing function of $r_{h-1}$ when $\ve_h$ is kept fixed. This follows from the fact that
\begin{align}
\frac{1}{\ol\ve_{h}}\frac{\partial r_h}{\partial r_{h-1}}&=\frac{\partial \vp(0|r_{h-1},\ve_{h},0)}{\partial r_{h-1}}\nonumber\\&\propto \frac{-\big( \frac{1}{1-r_{h-1}}\big)^2}{\Big[1+\frac{\alpha_0}{\beta}\big(\sum_{l=0}^{m_{h-1}-1}\frac{\alpha^l}{\beta^l}\big)\Big]^2} <0.
 \end{align}
An easy way to understand this behavior is to notice that $\alpha,\alpha_0$ increase with $r_{h-1}$, while $\beta$ decreases with $r_{h-1}$. Therefore, from (\ref{FB-R}), it follows that
\begin{align}
r_{h}^a=r_{h}^b&\Rightarrow\,\,\vp(0| r_{h-1}^a,\ve_{i+1},0)=\vp(0| r_{h-1}^b,\ve_{i+1},0) \nonumber\\
 &\Rightarrow\,\, r_{h-1}^a=r_{h-1}^b. \label{Eqn-requiv}
\end{align}
Now, from (\ref{Eqn-requiv}) and (\ref{FB-PB}) guarantee ${p_b}_{h-1}^a={p_b}_{h-1}^b$. We then use the monotonicity of $\vp(0|r_{h-2},\ve_{h-1},{p_b}_{h-1})$ in conjunction with already shown results to show that $r_{h-2}^a=r_{h-2}^b$ and ${p_b}_{h-2}^a={p_b}_{h-2}^b$. Extending this inductively, we have $\mcr^a=\mcr^b$ and $\mcp^a=\mcp^b$. Therefore, for each $\delta>0$, there is at most one solution satisfying $r_{h}=\delta$.

Now, consider two rate-approximate solutions $(\mcr^a,\mcp^a)$ and $(\mcr^b,\mcp^b)$ such that
$0<\delta_a=r_{h}^a<r_{h}^b=\delta_b<1$. By monotonicity of $\vp(0|r_{h-1},\ve_{h},0)$, we have $r_{h-1}^a<r_{h-1}^b$. From (\ref{FB-PB}), we notice that ${p_b}_i$ is also a strictly increasing function in both its variables $r_i$ and ${p_b }_{i+1}$. Therefore, ${p_b}_{h-1}^a<{p_b}_{h-1}^b$. Again, proceeding inductively from the last node to the first each time noticing the monotonic growth of $(\ref{FB-PB})$ and $(\ref{FB-R})$, we conclude that
\begin{equation}
\begin{array}{lll}
r_i^a &<& r_i^b \\
{p_b}_i^a &<& {p_b}_i^b \\
\end{array}, \quad\quad i=1,\ldots,h. \label{Ap1-Contra}
\end{equation}
However, since $(\mcr^a,\mcp^a)$ and $(\mcr^b,\mcp^b)$ are both rate-approximate solutions, we have $r_1^a=r_1^b=\ol\ve_1$, which contradicts (\ref{Ap1-Contra}). Therefore, there is at most one solution to the system of equations.

To identify the unique solution, we construct a sequence of tuples $\{(\mcr[l],\mcp[l])\}_{l\in\mathbbm{N}}$ as described in Algorithm~\ref{alg:RbIE}. Note that Step 2 of the algorithm can be replaced by a convergence-type step that halts if $\|\mathbf{r}[l]-\mathbf{r}[l-1]\|_1$ is smaller than a chosen threshold.

\begin{algorithm}[h!]
\small{\caption{\emph{Rate-based Iterative Estimate}}\label{alg:RbIE}
\begin{algorithmic}[1]
\STATE $\texttt{Count}=1$ and ${p_b}_i[\texttt{Count}]=0$, $i=1,\ldots,h-1$.
\WHILE {\texttt{Count}$\leq$\texttt{Max\_Iter}}
\STATE ${p_b}_h[\texttt{Count}]=0$, $r_1[\texttt{Count}]={1-\ve_1}$, and $j=1$.
\WHILE {$j<h$}
\STATE Compute $r_{j+1}[\texttt{Count}]$, ${p_b}_j[\texttt{Count}+1]$ employing (\ref{FB-R}) and (\ref{FB-PB}) (that use $r_{j}[\texttt{Count}]$, ${p_b}_{j+1}[\texttt{Count}]$)
 \STATE $j\leftarrow j+1$.
\ENDWHILE
\STATE $\texttt{Count}\leftarrow\texttt{Count}+1$.
\ENDWHILE
\end{algorithmic}}
\end{algorithm}

By the monotonic property of the non-linear system of equations, the following results can be established.
\begin{equation}
\begin{array}{lll}
r_i[l] &<& r_i[l+1] \\
{p_b}_i[l] &<& {p_b}_i[l+1] \\
\end{array}, \quad\quad l\in\mathbbm{N}. \label{Ap1-Mont}
\end{equation}
However, each component of $\mcr$ and $\mcp$ is individually bounded by unity. Therefore, the sequence of numbers for each component of these vectors must converge. Denote the component-wise limit as $\mathbf{W}^*=(\mcr^*,\mcp^*)$. Denote $\Xi:[0,1]^{h}\times[0,1]^{h}\longrightarrow [0,1]^{h}\times[0,1]^{h}$ to be the following map. For each $\mcr,\mcp\in[0,1]^h$, denote $\Xi(\mcr,\mcp)$ to be the pair, whose first component is the vector of rates computed from
(\ref{FB-R}) and the second component is the vector of blocking probabilities computed from (\ref{FB-PB}). Then, $\Xi$ is a continuous map and $\Xi((\mcr[l],\mcp[l]))=(\mcr[l+1],\mcp[l+1])$ for each $l\in\mathbbm{N}$. Also, for this sequence of rates and blocking probabilities, we note that
\begin{align}
\hspace{-2mm}\|\Xi(\mathbf{W}^*)-\mathbf{W}^*\|_\infty&\leq \|(\mcr[l],\mcp[l])-\mathbf{W}^*\|_\infty \nonumber\\
 &+ \|\Xi((\mcr[l],\mcp[l]))-(\mcr[l],\mcp[l])\|_\infty \label{Ap1-conv}\\
& + \|\Xi(\mathbf{W}^*)-\Xi((\mcr[l],\mcp[l]))\|_\infty. \nonumber
\end{align}
However, the right-hand side of (\ref{Ap1-conv}) is true for any $l\in\mathbb{N}$. By allowing $l\rightarrow \infty$, the three limits vanish and hence we see that $\mathbf{W}^*=(\mcr^*,\mcp^*)$ is a fixed point of the map and hence the unique solution to the system of non-linear equations.

Finally, to see the conservation of flow, notice that the Rate-based Iterative Estimate models the system using a discrete-time M/M/1/$k$ system by the introduction of additional assumptions and parameters. In the model, the number of innovative packets that are successfully stored by $v_i$ as the system progresses from $l=0$ to $l=N$ is given by
\begin{align}
&Nr_i^*\big(1-\Pr[n_i=m_i]+\ol\ve_{i+1}\ol{p_b}_{i+1}^*\Pr[n_i=m_i]\big)+o(N)\nonumber\\
&= Nr_i^*\big(1-\varphi(m_i|r_i^*,\ve_{i+1},{p_b}_{i+1}^*)\big(1-\ol\ve_{i+1}^*\ol{p_b}_{i+1}^*\big)\big)+o(N)\nonumber\\
&\hspace{-1.5mm}\stackrel{(\ref{FB-PB})}{=} Nr_i^*\ol{p_b}_{i}^*+o(N)\label{FB-Cons1}.
\end{align}
Similarly, the number of packets successfully output by $v_i$ is given by
\begin{align}
&N\ol\ve_{i+1}^*\big(1-\Pr[n_i=0]\big)\ol{p_b}^*_{i+1}+o(N)\nonumber\\
&=N\ol\ve_{i+1}^*\big(1-\varphi(0|r_i^*,\ve_{i+1},{p_b}_{i+1}^*)\big)\ol{p_b}^*_{i+1}+o(N)\nonumber\\ &\hspace{-1.5mm}\stackrel{(\ref{FB-R})}{=}Nr_{i+1}^*\ol{p_b}^*_{i+1}+o(N)\label{FB-Cons2}
\end{align}
Since the M/M/1/$k$ system is lossless, all stored packets eventually leave the system. Thus, the average rate of packet storage at a node must match the average rate of packets output from that node. Comparing (\ref{FB-Cons1}) with (\ref{FB-Cons2}), the conservation of packet flow for the rate-approximate solution follows.

\section{Proof of Theorem~\ref{FB-LBthm}}\label{App.2}

The proof elaborates the behavior of a tandem system via a formal setup for the discrete-time equivalent of the $G/M/1/k$ queue~\cite{SKBBook}. To illustrate the complications in the setup, Fig.~\ref{App2-Fig1} presents a section of an inter-arrival period for the first node. The number of customers in the queue of the node just before an arrival or a departure is presented on the axis. The arrival and departure of customers is marked by incoming and outgoing arrows, respectively. In Scenario A, we see that the queue is never starved and as a result all the inter-departure times are instances of the service process.

\begin{figure}[h]
\psfrag{a}{4} \psfrag{b}{5} \psfrag{c}{4} \psfrag{d}{3} \psfrag{e}{2} \psfrag{f}{1} \psfrag{g}{0}
\psfrag{X}{$X$} \psfrag{A}{\hspace{-8mm}Scenario A} \psfrag{B}{\hspace{-8mm}Scenario B} \psfrag{i}{1}
\centering
\includegraphics[width=3.4in, angle=0]{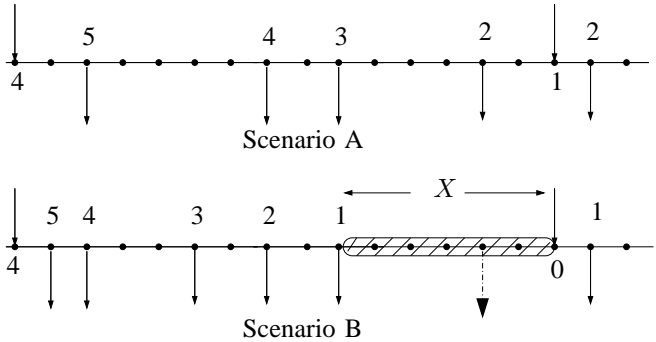}
\caption{A section of inter-arrival periods at the first server (assuming it possesses five customer slots).}
 \label{App2-Fig1}
\end{figure}

However, in Scenario B, we notice that all the five customers that are in the queue after the arrival are serviced much ahead of the next arrival and hence there is a period of time during which the queue is starved. If the queue were not starved, it could have possibly serviced a customer at the instance marked by the outgoing dotted arrow. Hence, this duration of time denoted by $X$ in the figure, adds a delay to the inter-departure time. Thus, if we are able to extract the distribution
$\{f^X(i)\}_{i\in\mathbbm{N}}$ of this duration, we can identify the inter-arrival distribution $g^\textrm{out}$ as seen by the second node to be a weighted sum of $f_X\otimes\mathbbm{G}(\tee_N)$ and $\mathbbm{G}(\tee_N)$.

In order to identify the distribution $f_X$, we need to identify the probability distribution $\pi$ of the number of customers in the first node's buffer just after an arrival. The first step in identifying $\pi$ from the imbedded Markov chain for the occupancy of the first node is to construct the distribution $\{D_j\}_{j\in\mathbbm{Z}_{\geq 0}}$ of the number of packets that could be potentially transmitted
during an inter-arrival duration $T_A$ provided the queue were infinite. This distribution can be computed from the arrival and departure processes in the following manner.
\begin{align}
D_j&=\sum_{k=1}^\infty \Pr[T_A=k]\binom{k}{j}\tilde{\tee}_{N}^{k-j}\ol{\tilde\tee}_{N}^j \nonumber\\&=\sum_{k=1}^\infty \bigl(\sum_{l=1}^{N-1} p_l \ol{\tee}_l\tee_l^{k-1}\bigr)\binom{k}{j}\tilde\tee_{N}^{k-j}\ol{\tilde\tee}_{N}^{j}\nonumber\\
&=\sum_{l=1}^{N-1} p_l\frac{\ol{\tee}_l}{\tee_l}\Big[\frac{\ol{\tilde\tee_N}}{\tilde\tee_N}\Big]^j\Bigl(\sum_{k=1}^{\infty} \binom{k}{j}(\tee_l\tilde\tee_N)^k\Bigr)\nonumber\\
&\stackrel{(a)}{=}\Big[\frac{\ol{\tilde\tee_N}}{\tilde\tee_N}\Big]^j\sum_{l=1}^{N-1} \frac{p_l\ol{\tee}_l}{\tee_l}\Big(\frac{(\tee_l\tilde\tee_N)^j}{(1-\tee_l\tilde\tee_N)^{j+1}}-\sigma[1-j]\Big),
\end{align}
where in the above, we use $\tilde\tee_N\triangleq\tee_N+\ol\tee_N(1-q)$ to incorporate the actual parameter of the memoryless service time, and in (a), we use $\frac{1}{(1-x)^{n+1}}=\sum_{r\geq 0}\binom{r}{n}x^{r-n}$, $0<|x|<1$. For each $i,j\in\{1,\ldots, m\}$, the $(i,j)^{\textrm{th}}$ entry of the probability transition matrix $P_\pi$ for the imbedded Markov chain that tracks the number of customers just after an arrival can be computed by
\begin{align}
(P_\pi)_{i,j}=\sigma[2-j]\Bigl(\sum_{k=i}^\infty D_k\Bigr)&+\sigma[j-1]D_{i+1-j}\nonumber\\&+\sigma[j-m+1]D_{i-j}.\label{App2-eqn1}
\end{align}
Note that in (\ref{App2-eqn1}), we set $D_{k}=0$ when $k<0$. The distribution $\pi$ can then be solved from the eigenvector relation $ \pi(I-P_\pi)=\mathbf{0}$. Note that a packet arriving at the first node will not be accepted if the node is in full buffer and no packet had left in the preceding inter-arrival duration. The probability of this blocking event at the first node is given by
\begin{equation}
\mathcal{P}(g^{\textrm{in}},m_i,\tee_N,q)\triangleq\pi_mD_0. \label{FB-BlockingEqn}
\end{equation}

Finally, we can identify the distribution of $X$ by conditioning on the number of customers $M$ just after a customer arrival. It is seen that for $i,k>0$,
\begin{align}
\Pr[X=i | M=k ] &= \sum_{j=1}^\infty\hspace{-0.5mm} \bigg[\hspace{-1.5mm}\begin{array}{ll}\Pr[\textrm{the queue is emptied at time $j$}]\\ \times \Pr[T_A=i+j]\end{array}\hspace{-1.5mm}\bigg] \nonumber\\
&= \sum_{j=1}^\infty\binom{j-1}{k-1}\ol{\tilde\tee}_N^k{\tilde\tee_N}^{j-k}\Bigl[\sum_{l=1}^{N-1} p_i \ol{\tee}_l\tee_l^{i+j-1}\Bigr]\nonumber\\
&= \sum_{l=1}^{N-1}p_l\ol\tee_l\tee_l^{i-1}\Bigl[\frac{\ol{\tilde\tee}_N^k}{{\tilde\tee_N}^k}\sum_{j=1}^\infty \binom{j-1}{k-1}(\tee_l{\tilde\tee_N})^j\Bigr]\nonumber\\
&= \sum_{l=1}^{N-1}\Bigl(p_l\frac{(\tee_l\ol{\tilde\tee}_N)^k}{(1-\tee_l{\tilde\tee_N})^k}\Bigr)(\ol\tee_l\tee_l^{i-1}). \label{App2-eqn3}
\end{align}
From (\ref{App2-eqn3}), we notice that the distribution of $X$ conditioned on $M=k$ is a weighted sum of geometric distributions. The distribution of $X$ can then be computed as follows.
\begin{align}
f^X_i&=\frac{\sum_{k=0}^m \pi_k \Pr[X=i | M=k ]}{\sum_{k=0}^m \pi_k \Pr[X\geq 1|M=k]} =\sum_{l=1}^l \beta_l \ol\tee_l\tee_l^{i-1},\nonumber\\
\beta_l&= p_l \bigg[{\mathop{\sum_{k\in\{0,\ldots,m\}}}_{l=\{1,\ldots,N-1\}}} \hspace{-1mm} p_l\frac{\pi_k(\tee_l\ol{\tilde\tee_N})^k}{(1-\tee_l{\tilde\tee_N})^k}\bigg]^{-1}{{\sum_{0\leq k \leq m}}} \frac{\pi_k(\tee_l\ol{\tilde\tee_N})^k}{(1-\tee_l{\tilde\tee_N})^k}\nonumber
\end{align}
Also, we notice that the distribution of inter-arrival times $g^{\textrm{out}}$ as seen by the second node is either an instance of $f^X\otimes\mathbbm{G}(\tee_N)$ or that of $\mathbbm{G}(\tee_N)$, and hence can be written as
\begin{equation}
 g^{\textrm{out}}\triangleq\underbrace{\big(\alpha f^X+(1-\alpha)\mathbbm{I}\big)}_{\triangleq \Upsilon(g^{\textrm{in}},m,\tee_N,q)}\otimes\mathbbm{G}(\tee_N)
\end{equation}
for some $\alpha\in[0,1]$. The last step in constructing the inter-departure distribution is to identify $\alpha$. This is
done by noticing the mean duration between departures. Over a large duration $N$, the number of packets that are accepted at the first node is given by $\frac{N}{\langle g^{\textrm{in}}\rangle}(1-\mathcal{P}(g^{\textrm{in}},m_i,\tee_N,q)) +o(N)$. The number of packets that are accepted by the second node is given by $\frac{N}{\alpha\langle f^X\rangle+\frac{\ol\alpha}{1-\tee_N}}(1-q)+o(N)$. Since the system has finite buffer size and no loss, the rates  must match. Therefore, one can identify $\alpha$ using the following.
\begin{equation}
\frac{1}{\alpha\langle f^X\rangle+\frac{\ol\alpha}{1-\tee_N}}(1-q)=\frac{1}{\langle g^{\textrm{in}}\rangle}(1-\mathcal{P}(g^{\textrm{in}},m_i,\tee_N,q)).
\end{equation}
Finally, notice that if $\mu\neq\lambda$, we have
\begin{equation}
\mathbbm{G}(\lambda)\otimes\mathbbm{G}(\mu)=\frac{1-\lambda}{\mu-\lambda}\mathbbm{G}(\mu)+\frac{1-\mu}{\lambda-\mu}\mathbbm{G}(\lambda).
\end{equation}
Using the above we can see that
\begin{eqnarray}
g^{\textrm{out}}=\sum_{l=1}^{N-1}\frac{\alpha \beta_l\ol\tee_N}{\tee_l-\tee_N}\mathbbm{G}(\tee_l)+\Big(\ol\alpha+\sum_{l=1}^{N-1}\frac{\alpha \beta_l\ol\tee_l}{\tee_N-\tee_l}\Big)\mathbbm{G}(\tee_N),
\end{eqnarray}
which is also a weighted sum of geometric distributions.

\section{Proof of Theorem~\ref{Thm-NFOpt}}\label{App.3}
We present below a fundamental result that will be used in various stages of the proof.
\begin{lem}\label{RLCinnolem}
Let $X$ be a vector space over a finite field $\f_q$ and let $\lss(A)\triangleq \lsp(A)$ for any $A\subseteq X$. Let $U=\{u_1,\ldots,u_k\}$ and $V=\{v_1,\ldots,v_{k'}\}$ be two subsets such that $\lss(V)\subsetneq\lss(U\cup V)$. Then, let $\mathbf{a}$ be selected uniformly at random from $\f^{k}_q$ and set $V'=V\cup\{\sum_{j=1}^k a_i u_i\}$. Then,
\begin{equation}
\hspace{-2.5mm}\Pr\Big[\dim\big(\lss(V)\big) \hspace{-0.75mm}\nless \hspace{-0.75mm} \dim\big(\lss(V')\big)\Big]\hspace{-0.75mm}<\hspace{-0.75mm}\frac{q^{\dim(\lss(U)\cap\lss(V))}}{q^{\dim(\lss(U))}}.
\end{equation}
\end{lem}
\begin{proof}
Let $G_0$ be the set of all vectors $\mathbf{b}\in \f^k$ such that $\sum_{j=1}^k b_i u_i=0$. Then $G_0$ forms a commutative group under componentwise addition. Similarly, let for each $u\in\lss(U)$, let $G_u$ be the set of vectors $\mathbf{b}\in \f^k$ such that $\sum_{j=1}^k b_i u_i=u$. It is follows that $\{G_u:u\in\lss(U)\}\cong \f^k/G_0$, {i.e.}, they are the coset translates of the subgroup $G_0$. Therefore, uniform selection of the coefficients to perform a linear combination results in the selection of a vector in $\lss(U)$ uniformly at random. Notice that $\dim(\lss(V'))=\dim(\lss(V))$ if and only if $\sum_{j=1}^k a_i u_i\in \lss(U)\cap\lss(V)$. Note that the occurrence of this event is improbable for large fields, since
\begin{equation}
\Pr\big[\dim(\lss(V'))=\dim(\lss(V))+1\big]=1-\frac{q^{\dim(\lss(U)\cap\lss(V))}}{q^{\dim(\lss(U))}}.\nonumber
\end{equation}
\end{proof}
\begin{cor}\label{LinnolemCor}
Let $A$ be a $k\times n$ matrix with entries from $\f_q$ such that $\rk(A)=r$. Let $\mathbf{b}\in \f_q^n$ be selected uniformly at random. Then,
\begin{equation}
\Pr[A\mathbf{b}^T=\mathbf{0}]=q^{-r}.
\end{equation}
\end{cor}

The basic idea of the proof is to construct a chain for the setting without feedback that is similar to the EMC. Once the chain is identified, the proof will be completed by showing that the transition probabilities of each transition approaches that of the EMC as the field size is made large. To this end, allow $M_i(l)$ to be the packet received by the node $v_i$ at the $l^\textrm{th}$ epoch. Set $M_i(l)=0$ if the $(v_{i-1},v_i)$ channel erases the transmitted packet at the $l^\textrm{th}$ epoch. For the sake of proof, each epoch is divided into $h$ sub-epochs. Since the network is assumed to work in a transmit-first mode, the network updates the buffers in a reverse-hop fashion, {i.e.}, at the $j^{\textrm{th}}$ sub-epoch, the message generated by $v_{h-j}$ is used to update that of $v_{h-j+1}$. Define $B_i(l,j)=\{P_{i,k}(l,j):k=1,\ldots, m_i\}$ to be the set of packets in the buffer of $v_i$ after the $j^{\textrm{th}}$ sub-epoch of the $l^{\textrm{th}}$ epoch. For notational ease, let $\mcW_i(l,j)\triangleq\lsp\Big(\bigcup\limits_{h\geq i'\geq i} B_{i'}(l,j)\Big)$ for $i\in\{1,\ldots, h\}$ and $j\in\{1,\ldots, h\}$. Note that the system is uniquely described by the dynamics of the nested vector spaces $\big\{\mc{W}_i(l,h)\big\}_{i=1}^{h}$. Define occupancy for this coded setting as
\begin{equation}
\hspace{-1.5mm}\eta_i(l,j)=\dim(\mcW_i(l,j))-\dim(\mcW_{i+1}(l,j)), \hspace{-0.5mm}\begin{array}{l}1\leq i<h\\  1\leq j\leq h\end{array}\hspace{-1.5mm}.
\end{equation}
Notice that this notion of occupancy denotes the number of additional \emph{innovative} packets that is housed by $v_i$ at the $l^\textrm{th}$ epoch that has not been conveyed to downstream nodes. Also note that at the $j^{\textrm{th}}$ sub-epoch of the $l^{\textrm{th}}$ epoch, the only buffer that changes is that of $v_{h-j+1}$ due to the receipt of $M_{h-j}(l)$. Thus,
\begin{align}
\begin{array}{l}
\mcW_{i}(l,j)=\mcW_{i}(l,j-1),\\
\eta_i(l,j)=\eta_i(l,j-1),
\end{array} \hspace{-2mm}
\begin{array}{l}
  1\leq i\leq h,\, i\neq h-j+1\\
  1\leq i< h,\, i\neq h-j, h-j+1
\end{array} \nonumber
\end{align}
To investigate the change of occupancy\footnote{In accordance with the notation of (\ref{FB-eqn3}) and (\ref{AMCeqn}), occupancies $n_i(l+1)$ and $\nt_i(l+1)$ correspond to $\eta_i(l,h)$ in this setup.} after the $j^{\textrm{th}}$ sub-epoch of the $l^{\textrm{th}}$ epoch for nodes $v_{h-j}$ and $v_{h-j+1}$, we have to consider the following cases.
\begin{itemize}
\item[1.] \underline{$j=1$:} If $\eta_{h-1}(l-1,h)>0$, then by Lemma~\ref{RLCinnolem}, we see that
    \begin{align}
    \dim(\mcW_{h}(l,1))&=\dim(\lsp(\mcW_{h}(l-1,h)\cup\{M_{h}(l)\}))\nonumber\\&=\dim(\mcW_{h}(l-1,h))+1\nonumber
    \end{align}
     with probability at least $1-\frac{1}{q}$. Therefore, it follows that $\eta_{h-1}(l,1)=\eta_{h-1}(l-1,h)-1$ with high probability\footnote{Throughout this section, by `with high probability' we mean that we can guarantee any probability close to unity by choosing a large field size $q$.}. If on the other hand $\eta_{h-1}(l-1,h)=0$, then $\eta_{h-1}(l,1)=0$ and $\dim(\mcW_{h}(l,1))=\dim(\mcW_{h}(l-1,h))$.
\item[2.] \underline{$j>1$, $X_{h-j+1}(l)=0$:} In this case, $M_{h-j+1}(l-1)=0$ and there is no update at the buffers of the node $v_{h-j+1}$. Therefore,
    \begin{align}
    \begin{array}{lcl}
    \mcW_{h-j+1}(l,j)&=&\mcW_{h-j+1}(l,j-1)\\
    \eta_{h-j}(l,j)&=&\eta_{h-j}(l,j-1)\\
    \eta_{h-j+1}(l,j)&=&\eta_{h-j+1}(l,j-1)\\
    \end{array}.
    \end{align}
\item[3.] \underline{$j>1$, $X_{h-j+1}(l)=1$, $\eta_{h-j+1}(l,j-1)< m_{h-j+1}$ and }\\\underline{$\eta_{h-j}(l,j-1)\geq 0$:} Notice that in this case, since the occupancy of $v_{h-j+1}$ before the update is not full, there exists $\mathbf{a} \in \f_q^{m_{h-j+1}}\setminus \{\mathbf{0}\}$ such that
 \begin{equation}
\sum_{k=1}^{m_{h-j+1}} a_k P_{h-j+1,k}(l,j-1)\in \underbrace{\mcW_{h-j+2}(l,j-1)}_{(=\mcW_{h-j+2}(l,j))}.
\end{equation}
Suppose that $\mathbf{b}\in\f_q^{m_{h-j+1}}$ is used to update $B_{h-j+1}(l,j-1)$ with the message $M_{h-j+1}(l)$. Then,
\begin{align}
&\sum_{k=1}^{m_{h-j+1}} a_k{P_{h-j+1,k}(l,j)} \in \underbrace{\mcW_{h-j+1}(l,j)}_{(\supseteq \mcW_{h-j+2}(l,j))}\nonumber \\
&\Leftrightarrow {\begin{array}{l}\Big[{\Sigma}_{k} a_kb_k) M_{h-j+1}(l)\\\,\,+\underbrace{\Sigma_k a_kP_{h-j+1,k}(l,j-1)}_{\in \mcW_{h-j+2}(l,j) \subseteq \mcW_{h-j+1}(l,j)}\Big]\end{array}\in \mcW_{h-j+1}(l,j)}\nonumber\\
&\Leftrightarrow (\Sigma_k a_kb_k) M_{h-j+1}(l) \in \mcW_{h-j+1}(l,j)\nonumber
\end{align}

Now, note that the vector $\mathbf{a}$ is not unique and that the set of all vectors that relate the contents of $v_{h-j+1}$ form a vector space over $\f_q$ of dimension $m_{h-j+1}-\eta_{h-j+1}(l,j-1)$. Let $A$ be the matrix generated by enlisting all the vectors in this space as rows. Then $rank(A)=m_{h-j+1}-\eta_{h-j+1}(l,j-1)$ and by Cor.~\ref{LinnolemCor}, $\Pr[A\mathbf{b}^T = \mathbf{0}]=\frac{1}{q^{m_{h-j+1}-\eta_{h-j+1}(l,j-1)}}$. Therefore, by choosing a large field size, the probability of the event $M_{h-j+1}(l)\in\mcW_{h-j+1}(l,j)$ can be made arbitrarily close to unity. Finally from Lemma~\ref{RLCinnolem}, we see that if $\eta_{h-j}(l,j-1)>0$, then $M_{h-j+1}(l)$ is innovative w.h.p. It is straightforward to see that in this setting, if $\eta_{h-j}(l,j-1)>0$, an innovative packet is conveyed w.h.p. from $v_{h-j}$ to $v_{h-j+1}$ and
\begin{align}
\begin{array}{rcl}
\eta_{h-j+1}(l,j)&=&\eta_{h-j+1}(l,j-1)+1 \\
\eta_{h-j}(l,j)&=&\eta_{h-j}(l,j-1)-1
\end{array}
\end{align}
Finally, if $\eta_{h-j}(l,j-1)=0$, we notice that both occupancies remain unaltered w.h.p.
\item[4.] \underline{$j>1$, $X_{h-j+1}(l)=1$, $\eta_{h-j+1}(l,j-1)= m_{h-j+1}$ and}\\ \underline{$\eta_{h-j}(l,j-1)> 0$:} Suppose that by updating the buffers of $v_{h-j+1}$ with $M_{h-j+1}(l)$ (using a randomly selected $\mathbf{b}$), we introduce a linear dependency in the newly formed buffer entries. That is, $\exists \,\mathbf{a}\in\f_q^{m_{h-j+1}}\setminus\{\mathbf{0}\}$ such that
\begin{align}
&\Sigma_{k} a_kP_{h-j+1,k}(l,j)\in \underbrace{\mcW_{h-j+2}(l,j)}_{(=\mcW_{h-j+2}(l,j-1))} \nonumber\\
&\Leftrightarrow {\begin{array}{l}\Big[(\Sigma_k a_kb_k) M_{h-j+1}(l) \\\,\,+\underbrace{\Sigma_k a_kP_{h-j+1,k}(l,j-1)}_{\in \mcW_{h-j+1}(l,j-1)}\Big]\end{array}}\in \underbrace{\mcW_{h-j+2}(l,j-1)}_{\subseteq \mcW_{h-j+1}(l,j-1)}\nonumber\\
&\Leftrightarrow (\Sigma_k a_kb_k) M_{h-j+1}(l) \in \mcW_{h-j+1}(l,j-1)\nonumber\\
&\Leftrightarrow(\Sigma_k a_kb_k=0)\vee\big(M_{h-j+1}(l) \in \mcW_{h-j+1}(l,j-1)\big)\nonumber
\end{align}
However, from Lemma~\ref{RLCinnolem} and Cor.~\ref{LinnolemCor}, $\Pr[M_{h-j+1}(l) \in \mcW_{h-j+1}(l,j-1)]<O(\frac{1}{q})$ and $\Pr[\sum_{k}a_kb_k=0]=\frac{1}{q}$. Therefore, w.h.p. there is no linear dependency introduced after update and the occupancy is unaltered in this case.
\item[5.] \underline{$j>1$, $X_{h-j+1}(l)=1$, $\eta_{h-j+1}(l,j-1)= m_{h-j+1}$ and}\\\underline{$\eta_{h-j}(l,j-1)= 0$:} Just like before, the aim here is to show that there will be no change in occupancy. Since in this case, $v_{h-j}$ has no innovative packets, the message it generates will be a linear combination of packets in $v_{k}$, $k>h-j$. Therefore, we can write $M_{h-j+1}(l)=\sum_{k}e_k P_{h-j+1,k}(l,j-1)+W$, where $W\in\mcW_{h-j+2}(l,j-1)$ and $\mathbf{e}\in\f_q^{m_{h-j+1}}$. Let $\mathbf{b}\in\f_q^{m_{h-j+1}}$ be used to update the buffer of $v_{h-j+1}$ and let $\mathbf{a}\in\f_q^{m_{h-j+1}}$, then,
\begin{align}
&\Sigma_k a_kP_{h-j+1,k}(l,j)\in W_{h-j+2}(l,j)\nonumber\\
&\Leftrightarrow \begin{array}{l}\Big[\Sigma_k a_kP_{h-j+1,k}(l,j-1)\\\,\, + (\Sigma_k a_k b_k)M_{h-j+1}(l)\Big]\end{array}{\in W_{h-j+2}(l,j)}\nonumber\\
&\Leftrightarrow {\Sigma_{k,l} \big[a_k + e_ka_l b_l\big]P_{h-j+1,k}(l,j-1) \in W_{h-j+2}(l,j)}\nonumber
\end{align}
Note that the above is true if only if $a_k + e_k(\sum_l a_l b_l)=0$ for $1\leq k\leq m_{h-j+1}$, since $\eta_{h-j+1}(l,j-1)=m_{h-j+1}$. Therefore, a linear dependency of stored packets arises if and only if there is a non-trivial solution for $(1+\mathbf{e}^T\mathbf{b})\mathbf{x}=0$, where $\mathbf{e}=[e_1,\ldots,e_{m_{h-j+1}}]$, and $\mathbf{b}=[b_1,\ldots,b_{m_{h-j+1}}]$. However, this occurs if and only if $\det\big(I+\mathbf{e}^T\mathbf{b}\big)\neq 0$. Finally, note that this determinant is zero if and only if $\mathbf{b}\mathbf{e}^T=-1$. However, this event occurs with probability $O(\frac{1}{q})$, since the vector $\mathbf{b}$ is chosen uniformly at random from $\f_q^{m_{h-j+1}}$. Therefore, w.h.p. there is no linear dependency induced in the contents of $v_{h-j+1}$ and the occupancy of $v_{h-j+1}$ remains $m_{h-j+1}$.
\end{itemize}
To summarize,
\begin{itemize}
\item[a.] Although the dynamics of the system are driven by the spaces $\{\mcW_i(l,j)\}_{i,l,j}$, the transitions and their probabilities depend only on $\{\eta_i(l,j)\}_{l\geq 0}$, $i=1,\ldots,h-1$, $j=1,\ldots,h$, and not the spaces as such. Therefore, the system can be equivalently modeled using just these occupancy vectors as states.
\item[b.] The transition probabilities for the chain given by occupancies $\{\eta_i(l,h)\}_{l\geq 0}$, $i=1,\ldots,h-1$ approach that of the EMC as the field size is made large.
\end{itemize}
 Finally, since the steady-state probability is a continuous function of the probability transition matrix, the steady-state probabilities of the chain for networks without feedback approaches that of the EMC, thereby guaranteeing that the throughput achieved by the random coding scheme over a line network without feedback is asymptotically the same as that of a line network with identical parameters and perfect feedback.

\end{document}